\documentclass{article}%
\usepackage{amsmath}
\usepackage{amsfonts}
\usepackage{amssymb}
\usepackage{graphicx}%
\addtocounter{footnote}{1}
\makeatletter
\@addtoreset{figure}{section}
\def\thefigure{\thesection.\@arabic\c@figure}
\def\fps@figure{h,t}
\@addtoreset{table}{bsection}
\def\thetable{\thesection.\@arabic\c@table}
\def\fps@table{h, t}
\@addtoreset{equation}{section}

\makeatother
\newtheorem{theorem}{Theorem}

\newtheorem{example}[theorem]{Example}
\newtheorem{examples}[theorem]{Examples}

\newtheorem{lemma}[theorem]{Lemma}

\newtheorem{proposition}[theorem]{Proposition}
\newtheorem{remark}[theorem]{Remark}

\newenvironment{proof}[1][Proof]{\noindent\textbf{#1.} }{\ \rule{0.5em}{0.5em}}

\begin{document}

\title{\textbf{Monopoles in arbitrary dimension}}
\author{Pablo D\'{\i}az$^{1}$ and Joan-Andreu L\'{a}zaro-Cam\'{\i}$^{2}\bigskip$\\$^{1,2}${\small Departamento de F\'{\i}sica Te\'{o}rica. Universidad de
Zaragoza.}\\{\small Pedro Cerbuna, 12. E-50009 Zaragoza, Spain.\medskip}\\$^{2}${\small Centre de Recerca Matem\`{a}tica.}\\{\small Apartat 50. E-08193 Bellaterra (Barcelona), Spain.}}
\date{}
\maketitle

\begin{abstract}
A self-contained study of monopole configurations of pure Yang-Mills theories
and a discussion of their charges is carried out in the language of principal
bundles. A $n$-dimensional monopole over the sphere $\mathbb{S}^{n}$ is a
particular type of principal connection on a principal bundle over a symmetric
space $K/H$ which is $K$-invariant, where $K=SO(n+1)$ and $H=SO(n)$. It is
shown that principal bundles over symmetric spaces admit a unique
$K$-invariant principal connection called canonical, which also satisfy
Yang-Mills equations. The geometrical framework enables us to describe their
associated field strengths in purely algebraic terms and compute the charge of
relevant (Yang-type) monopoles avoiding the use of coordinates. Besides, two
corrections on known results are performed in this paper. First, it is proven
that the Yang monopole should be considered a connection invariant by
$Spin(5)$ instead of by $SO(5)$, as Yang did in his original article
\cite{yang}. Second, unlike the way suggested in \cite{townsend}, we give the
correct characteristic class to be used to calculate the charge of the
$SO(2n)$-monopoles considered by Gibbons and Townsend.

\end{abstract}

\makeatletter\addtocounter{footnote}{1} \footnotetext{e-mail: \texttt{pdiaz}%
{\texttt{@unizar.es}}} \addtocounter{footnote}{1} \footnotetext{e-mail:
\texttt{lazaro}{\texttt{@unizar.es}}} \makeatother

\bigskip


\bigskip

\section{Introduction}

Monopoles in gauge theories have deserved a lot of attention since Dirac
introduced his magnetic monopole \cite{Dirac}, mainly due to the fact that
monopoles carry an intrinsically associated charge which only takes discrete
values, something that could easily explain the observable quantization of the
charge in electromagnetic theory. Recall that the Dirac monopole can be seen
as a static singular solution on $\mathbb{R}^{3}$ of a field theory with gauge
group $U\left(  1\right)  $. In practice, monopoles have never been observed,
and their existence is only justified from a theoretical point of view in
order to build a bridge between classical and quantum field theories. After
Dirac and the explosion of the popularity of gauge theories, there have been
other attempts to generalize the concept of monopole to different
(non-abelian) gauge groups in higher dimensions. Among them, \cite{yang} is
one of the most celebrated generalizations.

One of the most remarkable aspects of monopoles is that their charge is
related to the topological properties of the underlying space and strongly
depends on the way the gauge potential is attached to it. In other words,
monopoles cannot be understood \textit{at a local level} but their properties
need to be described form a \textit{global point of view.} In particular,
unless additional boundary conditions are required, there cannot exist
monopoles in the Euclidean space $\mathbb{R}^{n}$, $n\in\mathbb{N}$, but, on
the contrary, monopoles exhibit a singularity at the origin $0\in
\mathbb{R}^{n}$, where the charge is supposed to be. Therefore, $0\in
\mathbb{R}^{n}$ needs to be removed. On the other hand, it is widely known
that the theory of principal bundles provides the most satisfactory framework
to study and develop gauge theories from a geometrical (global) point of view
(see \cite{viallet}, \cite{bleecker}, and \cite{Eguchi}). Although the reader
is supposed to be familiar with the geometrical framework of gauge theories,
we are going to recall in this paper the main features of principal bundles
for the sake of a clearer exposition.

If we restrict to pure Yang-Mills theories, the framework of principal bundles
over $\mathbb{R}^{n}\backslash\{0\}$ seems to be the main mathematical tool to
tackle monopoles. However, the classification theory of principal bundles over
paracompact manifolds (\cite{milnor-a} and \cite{milnor-b}) requires in
general a rather sophisticated topological machinery that we would like to
avoid as much as possible. Since $\mathbb{R}^{n}\backslash\{0\}$ is homotopic
to $\mathbb{S}^{n-1}$, we can study principal bundles either over
$\mathbb{R}^{n}\backslash\{0\}$ or $\mathbb{S}^{n-1}$ indistinguishably as far
as the global properties of monopoles is concerned; for a given gauge group
$G$, principal bundles over $\mathbb{R}^{n}\backslash\{0\}$ and $\mathbb{S}%
^{n-1}$ are homomorphic and their structure can be recovered from one to the
other. Remember that two principal bundles are called homomorphic is there
exits a smooth map between them equivariant with respect to the actions of the
gauge group. The key point is that $\mathbb{S}^{n}$ is a homogeneous space;
for example, $\mathbb{S}^{n}\cong\left.  SO(n+1)\right/  SO(n)$, where $SO(n)$
denotes the special orthogonal group. Since such spaces and their associated
structures have been extensively studied, a huge geometrical machinery is
consequently available to deal with them.

Using a geometrical language, gauge potentials and field strengths in gauge
field theories are described in terms of principal connections on principal
bundles and their curvature, respectively. On the other hand, the Chern-Weil
homomorphism provides a mechanism to associate to the curvature some
\textit{de Rham} cohomology classes $H^{2k}(\mathbb{S}^{n})$ of even order,
known as characteristic classes. Roughly speaking, the Chern-Weil homomorphism
allows us to remove the dependence of the field strength on the \textit{gauge
indices} (or the \textit{color}, in a physics language), which should not
appear in any observable physical quantity. In this context, a monopole
configuration on $\mathbb{S}^{n}$ is a principal bundle $\pi:P\rightarrow
\mathbb{S}^{n}$ with a principal connection such that:

\begin{enumerate}
\item[\textbf{(i)}] There exists a characteristic class in $H^{n}\left(
\mathbb{S}^{n}\right)  $ whose integral over $\mathbb{S}^{n}$ is different
from zero. This means that we can associate a non-vanishing charge to the
monopole. As we will discuss in Section \ref{section Chern classes}, there is
no general consensus on which topological invariant should represent the
charge of a monopole and some authors chose others. Observe that $n$ needs to
be even in order to $n/2$ be an integer. That is, there will be no monopoles
in even (spatial) dimensions.

\item[\textbf{(ii)}] The principal connection is $SO(n+1)$ invariant. This
property is usually referred to as {\bfseries\itshape spherical symmetry} of
the monopole in the literature. In particular, it implies that we need to be
able to define an action of the group of rotations of $\mathbb{R}^{n+1}$ on
our principal bundle so that the principal connection is invariant with
respect to it. This is \textit{not} always possible, as it actually happens
for the Yang's monopole, despite the explicit reference to the $SO(5)$
invariance Yang did in \cite{yang}. We will see that, in the Yang case,
spherical symmetry needs to be implemented through an action of $Spin(5)$
instead of $SO(5)$, contrary to what was usually thought.
\end{enumerate}

It is customary in gauge theories to give monopole configurations locally on
coordinate patches and then to impose some compatibility conditions where
these patches overlap. The use of coordinates is sometimes unavoidable in
computations, but it is often very tedious. Fortunately, there are many
features that can be seen intrinsically. The purpose of our paper is to convey
to the physics community some of the global tools from differential geometry
perfectly tailored to study monopoles. The main contributions of this paper
are the following:

\begin{enumerate}
\item[\textbf{1.}] We explicitly show that there exists a bijective
correspondence between principle bundles over the Euclidean space
$\mathbb{R}^{2n+1}\backslash\{0\}$ and principal bundles over the sphere
$\mathbb{S}^{2n}$ and their principal connections are Yang-Mills if and only
if they are Yang-Mills on the latter.

\item[\textbf{2.}] We will see that on $\mathbb{S}^{2n}$, seen as a symmetric
space, only the so-called canonical connections are $SO(2n+1)$-invariant.
Moreover, it is proved (see Proposition \ref{prop canonical is YM}) that they
automatically satisfy the Yang-Mills equations.

\item[\textbf{3.}] Despite the widely spread idea that the Yang monopole on
$\mathbb{S}^{4}$ is $SO(5)$ invariant, it is shown that the concept of
spherical symmetry needs to be implemented by its universal covering group
$Spin(5)$. This is because there does not exist any principal bundle with
structural group $SU(2)$ admitting a (left) $SO(5)$ action. When describing
the monopole on $\mathbb{S}^{4}$ by means of local sections as Yang did,
$Spin(5)$ acts through $SO(5)$, which explains why such a confusion arises.

\item[\textbf{4.}] We make precise some of the results about monopole
configurations found in the literature. Explicitly, in Section
\ref{seccion ejemplos} we discuss that the charge of the monopoles over
$\mathbb{S}^{2n}$ with gauge group $SO(2n)$, $n>2$, recently reviewed in
(\cite{townsend}) can only be implemented through the so-called Euler class.
Although, broadly speaking, the main ideas behind Gibbons and Townsend
$SO(2n)$-monopoles do not differ too much from ours, the way they introduce
the field strength and the charge of the monopole is imprecise and leads them
to assert wrong statements. We fix this point by clarifying the way to define
properly these concepts in geometrical terms.

\item[\textbf{5.}] We give a depiction of monopoles on homogeneous symmetric
spaces only in algebraic terms (Section \ref{section algebraic setting}). More
concretely, if $\pi:P\rightarrow$ $K/H$ is a principal bundle over a symmetric
Lie space related to a monopole with gauge group $G$, $K$ and $H\subset K$ two
Lie groups, then a monopole is completely described in terms of the Lie
algebras $\mathfrak{k}$, $\mathfrak{h}$, and $\mathfrak{g}$. This simplifies a
lot the amount of manipulations needed to compute any relevant quantity
associated to monopoles (no local coordinates are needed) and, what is more
important, allows us to go from a geometrical framework to an algebraic one
which, in practice, makes quantities computable. For example, we show in
Section \ref{section algebraic setting} and \ref{section Chern classes} that
field strengths and Chern classes can be easily computed for monopoles without
much effort.

\item[\textbf{6.}] We clarify the structure of monopole configurations from a
geometrical point of view. This means that our approach is global as we try to
emphasize the intrinsic nature of the structures involved in such
configurations and, consequently, avoid using local coordinates. As we said,
this approach seems to be suitable since the properties of monopoles are topological.

\item[\textbf{7.}] We gather some results on principal bundles over
homogeneous spaces which have appeared since the late 1950's and make them
available to physicists interested in monopoles. Although they are widely
known among geometers, there still exists surprisingly some confusion in the
community about the precise meaning of some concepts such as spherically
symmetric potentials, for instance, or the relationship between the charge of
a monopole and the topological invariants of a principal bundle expressed by
the Chern-Weil homomorphism.
\end{enumerate}

The paper is structured as follows: in Section 2, we recall on the one hand
the main geometric tools of principal bundles emphasizing their importance in
gauge theories and, on the other, we proof that principal bundles over
$\mathbb{S}^{2n}$ and $\mathbb{R}^{2n+1}\backslash\{0\}$ can be recovered ones
from the others. In Section 3, we introduce homogeneous principal bundles
$P_{\lambda}\rightarrow K/H$ over homogeneous spaces. These bundles, which
admit a left action by the Lie group $K$, are the geometric background for
monopole configurations. We characterize the principal connections (gauge
potentials) $\omega\in\Omega^{1}\left(  P;\mathfrak{g}\right)  $ which are
invariant by $K$ and show that, when $K/H$ is a symmetric space, there exists
a unique connection with these properties. We present in Section 4 an explicit
procedure to give the spherically symmetric field strengths $\Omega^{\omega}$
associated to monopole configurations in terms of the Lie algebras of the
groups involved. This procedure is implemented in some examples. In Section 5,
we recall the Chern-Weil homomorphism, a mechanism to associate some
\textit{de Rham} cohomology classes of $\mathbb{S}^{2n}$ to the field strength
$\Omega^{\omega}$ of $P_{\lambda}\rightarrow\mathbb{S}^{2n}$. We also show how
to define the charge of a monopole from these classes using the algebraic
description of $\Omega^{\omega}$ given in Section 4. Finally, in Section 6, we
apply the tools developed throughout the paper to revise the classical
examples by Dirac and Yang, the $SO(2n)$-monopoles widely studied by T.
Tchrakian and recently reviewed by Gibbons and Townsend (\cite{townsend}), and
the more recent $SU(2^{n-1})$-monopoles introduced by G. Meng (\cite{meng}%
).\bigskip

\noindent\textbf{Notation:} All manifolds $M$ in this paper will be of class
$C^{\infty}$. The set of smooth vector fields on $M$ will be denoted by
$\mathfrak{X}\left(  M\right)  $ and the set of differential forms by
$\Omega\left(  M\right)  $. If $M$ and $N$ are two manifolds, the tangent map
of a smooth function $F:M\rightarrow N$ at a point $m\in M$ between the
tangent spaces $T_{m}M$ and $T_{F(m)}N$ of $M$ and $N$ at $m\in N$ and $F(m)$
respectively will be denoted by $T_{m}F$. The symbol $d$ will be reserved for
the exterior differential $d:\Omega\left(  M\right)  \rightarrow\Omega\left(
M\right)  $. If $V$ is a real vector space, $\Lambda\left(  V\right)
=\oplus_{k\geq0}\Lambda^{k}\left(  V\right)  $ will be the space of
multilinear alternating maps from $V$ to $\mathbb{R}$. On the other hand,
$S_{n}$ will denote the symmetric group of order $n\in\mathbb{N}$ and
$\left\vert \sigma\right\vert =\pm1$ the parity of a permutation $\sigma\in
S_{n}$. The wedge product of two forms $\alpha\in\Omega^{k}(M)$ and $\beta
\in\Omega^{l}\left(  M\right)  $ is defined as%
\[
(\alpha\wedge\beta)\left(  X_{1},...,X_{k+l}\right)  =\frac{1}{k!l!}%
\sum_{\sigma\in S_{k+l}}\left(  -1\right)  ^{\left\vert \sigma\right\vert
}\alpha\left(  X_{\sigma(1)},...,X_{\sigma(k)}\right)  \beta\left(
X_{\sigma(k+1)},...,X_{\sigma(k+l)}\right)  ,
\]
$\{X_{1},...,X_{k+l}\}\subset\mathfrak{X}\left(  M\right)  $, and the
differential $d\alpha$ satisfies%
\begin{align*}
d\alpha\left(  X_{1},...,X_{k+1}\right)   &  =\sum_{i=1}^{k}\left(  -1\right)
^{i+1}\alpha\left(  X_{1},...,\widehat{X_{i}},...,X_{k+1}\right)  +\\
&  \sum_{i<j}\left(  -1\right)  ^{i+j}\alpha\left(  \lbrack X_{i},X_{j}%
],X_{1},...,\widehat{X_{i}},...,\widehat{X_{j}},...,X_{k+1}\right)  .
\end{align*}
It is worth noticing that, in the literature, some authors sometimes use
different factors in these expressions.

\section{Geometric preliminaries}

We recalled in the introduction that principal bundles over $\mathbb{R}%
^{2n+1}\backslash\{0\}$ and $\mathbb{S}^{2n}$ are in a bijective
correspondence. In this section, we are going to give more details about how
this bijection works. The idea is to use it in subsequent sections to switch
from $\mathbb{R}^{2n+1}\backslash\{0\}$ to $\mathbb{S}^{2n}$ and take
advantage of the geometric tools available when $\mathbb{S}^{2n}$ is
considered as a homogeneous space. Moreover, we want to see that, if a
principal connection on $\mathbb{S}^{2n}$ satisfies the Yang-Mills equations,
so does the corresponding induced connection on $\mathbb{R}^{2n+1}%
\backslash\{0\}$. The rest of this section is devoted to recalling the basics
of gauge theories such as principal connections (Subsection
\ref{subsection principal connections}) and the Hodge operator (Subsection
\ref{subsection Hodge operator}). After introducing Yang-Mills connections, we
will conclude the section by seeing that a principal connection is Yang Mills
on $\mathbb{S}^{2n}$ if and only if it is Yang-Mills on the corresponding
bundle over $\mathbb{R}^{2n+1}\backslash\{0\}$ (Proposition
\ref{proposition YM}).

\subsection{Correspondence between principal bundles over $\mathbb{R}%
^{2n+1}\backslash\{0\}$ and $\mathbb{S}^{2n}$}

Let $\pi:P\rightarrow M$ be a principal bundle with structural group $G$ over
a manifold $M$ and right action $R:G\times P\rightarrow P$. Let
$f:N\rightarrow M$ a smooth function from a manifold $N$ to $M$. The pull-back
of $\pi$ by $f$ is a fiber bundle over $N$ defined as%
\begin{align*}
f^{\ast}\left(  P\right)   &  =\left\{  \left(  p,x\right)  \in P\times
N~|~\pi\left(  p\right)  =f(x)\right\} \\
\overline{\pi}  &  :f^{\ast}\left(  P\right)  \rightarrow N\text{,
\ }\overline{\pi}\left(  \left(  p,x\right)  \right)  =x.
\end{align*}
With the natural right action $\left(  p,x\right)  \cdot g=\left(
R_{g}(p),x\right)  $, $g\in G$, inherited from $\pi:P\rightarrow M$, it is
easy to verify that $\overline{\pi}:f^{\ast}\left(  P\right)  \rightarrow N$
is indeed a principal bundle. An important result is that, if
$f,h:N\rightarrow M$ are two homotopic smooth maps, then the pull-backs
$f^{\ast}\left(  P\right)  $ and $h^{\ast}\left(  P\right)  $ are isomorphic
(see \cite[page 121]{isham},\cite{morrison}), that is, there exists a map
$F:f^{\ast}\left(  P\right)  \rightarrow h^{\ast}\left(  P\right)  $ over the
identity on $N$ such that $F\left(  z\cdot g\right)  =F\left(  z\right)  \cdot
g$ for any $z\in f^{\ast}\left(  P\right)  $. This rather simple result allows
us to explicitly draw the bijection between principal bundles over
$\mathbb{R}^{2n+1}\backslash\{0\}$ and $\mathbb{S}^{2n}$, respectively.
Indeed, let $\pi:P\rightarrow\mathbb{R}^{2n+1}\backslash\{0\}$ be a principal
bundle and let $\left.  P\right\vert _{\mathbb{S}^{2n}}$ be the restriction of
$P$ to $\mathbb{S}^{2n}$, which coincides with the pull-back of $P$ by the
inclusion of the sphere into $\mathbb{R}^{2n+1}\backslash\{0\}$ (\cite[page
120]{isham}). On the other hand, the map%
\begin{equation}%
\begin{array}
[c]{rcl}%
f:\mathbb{R}^{2n+1}\backslash\{0\} & \longrightarrow & \mathbb{S}^{2n}%
\subset\mathbb{R}^{2n+1}\backslash\{0\}\\
x & \longmapsto & \frac{x}{\left\Vert x\right\Vert },
\end{array}
\label{eq proyeccion homotopica}%
\end{equation}
is homotopic to the identity $\operatorname*{Id}:\mathbb{R}^{2n+1}%
\backslash\{0\}\rightarrow\mathbb{R}^{2n+1}\backslash\{0\}$, where $\left\Vert
x\right\Vert =\sqrt{\sum_{i=1}^{2n+1}(x^{i})^{2}}$ denotes the Euclidean norm.
Therefore, the principal bundles $f^{\ast}\left(  P\right)  $ and $P$ are
isomorphic. But clearly $f^{\ast}\left(  P\right)  =f^{\ast}\left(  \left.
P\right\vert _{\mathbb{S}^{2n}}\right)  $. So we conclude that principal
bundle structures on $\mathbb{S}^{2n}$ are induced by restriction from those
on $\mathbb{R}^{2n+1}\backslash\{0\}$ and, conversely, that principal bundles
over $\mathbb{S}^{2n}$ induce principal bundles over $\mathbb{R}%
^{2n+1}\backslash\{0\}$ by means of (\ref{eq proyeccion homotopica}), both
procedures being commutative. As a consequence, we can study monopole
configurations on the sphere $\mathbb{S}^{2n}$ and then pull them back onto
$\mathbb{R}^{2n+1}\backslash\{0\}$ using the projection
(\ref{eq proyeccion homotopica}). Before that, we will continue recalling more
geometric ingredients of gauge theories; concepts that are quite common for
physicists in the context of Riemannian geometry but less known in more
general principal bundle framework.

\subsection{Principal connections\label{subsection principal connections}}

Let $\pi:P\rightarrow M$ be a principal bundle with structural group $G$. A
{\bfseries\itshape principal connection} $\omega\in\Omega^{1}\left(
P;\mathfrak{g}\right)  $ is a one form on $P$ with values in the Lie algebra
$\mathfrak{g}$ of $G$ such that%
\begin{subequations}
\begin{gather}
R_{g}^{\ast}\omega=\operatorname*{Ad}\nolimits_{g^{-1}}\omega
,\label{monopolo 18 a}\\
\omega(p)\left(  \left.  \frac{d}{dt}\right\vert _{t=0}R_{\exp(t\eta
)}(p)\right)  =\eta\label{monopolo 18 b}%
\end{gather}
for any $g\in G$, $p\in P$, and $\eta\in\mathfrak{g}$. In this expression
$\operatorname*{Ad}$ denotes the adjoint representation of $G$ on
$\mathfrak{g}$ and $\exp:\mathfrak{g}\rightarrow G$ the usual exponential map.
We will denote the vector field $\left.  \frac{d}{dt}\right\vert _{t=0}%
R_{\exp(t\eta)}$ simply by $\eta_{P}$, $\eta\in\mathfrak{g}$. Any principal
connection $\omega\in\Omega^{1}\left(  P;\mathfrak{g}\right)  $ defines the
{\bfseries\itshape horizontal space} $\operatorname*{Hor}_{p}=\ker\omega$ at
any $p\in P$ such that $T_{p}P=\operatorname*{Hor}_{p}\oplus
\operatorname*{Ver}_{p}$, where $\operatorname*{Ver}_{p}\subset T_{p}P$ is the
{\bfseries\itshape vertical space} $\operatorname*{Ver}_{p}=\ker T_{p}\pi$. An
arbitrary form is called {\bfseries\itshape horizontal} if it vanishes when
contracted with vector fields in the vertical space.

Given a $\mathfrak{g}$ valued $r$-form $\varphi\in\Omega^{r}\left(
P,\mathfrak{g}\right)  $ on a principal bundle $\pi:P\rightarrow M$ and a
principal connection $\omega\in\Omega^{1}\left(  P,\mathfrak{g}\right)  $, the
{\bfseries\itshape covariant derivative} $D^{\omega}\varphi$ of $\varphi$ is
defined as%
\end{subequations}
\[
D^{\omega}\varphi(X_{1},...,X_{r}):=d\varphi\left(  X_{1}^{H},...,X_{r}%
^{H}\right)  ,~~X_{1},...,X_{r}\in\mathfrak{X}\left(  P\right)  ,
\]
where, at any point $p\in P$, $X_{i}^{H}(p)\in\operatorname*{Hor}_{p}$ is the
horizontal part of $X_{i}$, $i=1,...,r$. That is, we calculate the standard
exterior differential of $\varphi$ and then we restrict it to the horizontal
space. In particular, the {\bfseries\itshape curvature} of the connection is
$\Omega^{\omega}:=D^{\omega}\omega$. When regarded as a potential, we will
usually refer to the curvature as the {\bfseries\itshape field strength}. It
is customary to find the curvature in the literature written as $\Omega
^{\omega}=d\omega+\frac{1}{2}[\omega,\omega]$. This is the so-called
{\bfseries\itshape structural equation}. If $\varphi\in\Omega^{r}\left(
P,\mathfrak{g}\right)  $ and $\psi\in\Omega^{k}\left(  P,\mathfrak{g}\right)
$, the bracket $[\cdot,\cdot]$ is defined as%
\[
\lbrack\varphi,\psi]\left(  X_{1},...,X_{r+k}\right)  =\frac{1}{r!k!}%
\sum_{\sigma\in S_{r+k}}\left(  -1\right)  ^{\left\vert \sigma\right\vert
}[\varphi\left(  X_{\sigma(1)},...,X_{\sigma(r)}\right)  ,\psi\left(
X_{\sigma(r+1)},...,X_{\sigma(r+k)}\right)  ]_{\mathfrak{g}}.
\]
In this equation, the bracket $[\cdot,\cdot]_{\mathfrak{g}}$ is that of the
Lie algebra, $X_{1},...,X_{r+k}\in\mathfrak{X}\left(  P\right)  $ are
arbitrary vector fields on $P$, and $S_{r+k}$ denotes the permutation group of
$r+k$ elements.

\subsection{The Hodge operator and Yang-Mills
connections\label{subsection Hodge operator}}

Given a principal bundle $\pi:P\rightarrow M$ with structural Lie group $G$,
the {\bfseries\itshape adjoint bundle $\operatorname*{Ad}\left(  P\right)  $
}is the associated bundle $P\times_{\operatorname*{Ad}}\mathfrak{g}$. That is,
the space of equivalent classes of $P\times\mathfrak{g}$ under the equivalence
relation $(p,\xi)\sim(R_{g}(p),\operatorname*{Ad}_{g^{-1}}\xi)$, $p\in P$,
$\xi\in\mathfrak{g}$, and $g\in G$. It is a rather standard result in
differential geometry (see \cite[Theorem 19.14]{michor}) that the space
$\Omega_{equiv}(P;\mathfrak{g})^{\operatorname*{Hor}}$ of horizontal
$\mathfrak{g}$-valued forms on $P$ which are $G$-equivariant by
(\ref{monopolo 18 a}) can be identified with the space $\Omega\left(
M;\operatorname*{Ad}\left(  P\right)  \right)  $ of $\operatorname*{Ad}\left(
P\right)  $-valued differential forms on the base manifold $M$. This
identification works as follows: having a (principal)\ connection $\omega
\in\Omega^{1}\left(  P;\mathfrak{g}\right)  $ amounts to having a splitting of
the exact short sequence%
\[
0\longrightarrow V_{p}\longrightarrow T_{p}P\overset{\curvearrowleft\Gamma
_{p}}{\longrightarrow}T_{\pi(p)}M
\]
at any point $p\in P$ such that $X=\Gamma_{p}\left(  T_{p}\pi(X)\right)
\in\operatorname*{Hor}\nolimits_{p}$ for any $X\in T_{p}P$. Thus, we naturally
associate to any $\varphi\in\Omega_{equiv}^{r}(P;\mathfrak{g}%
)^{\operatorname*{Hor}}$ the $\operatorname*{Ad}\left(  P\right)  $-valued
form $\widetilde{\varphi}\in\Omega^{r}\left(  M;\operatorname*{Ad}\left(
P\right)  \right)  $ such that%
\begin{equation}
\widetilde{\varphi}(m)\left(  Y_{1},...,Y_{r}\right)  =\left[  p,\varphi
(p)\left(  \Gamma_{p}\left(  Y_{1}\right)  ,...,\Gamma_{p}(Y_{r})\right)
\right]  ^{\sim}\label{monopole 10}%
\end{equation}
for any $Y_{1},...,Y_{r}\in\mathfrak{X}\left(  M\right)  $. In
(\ref{monopole 10}), $p\in\pi^{-1}(m)$, and the bracket $[\cdot,\cdot]^{\sim}$
denotes the equivalent class of a point $\left(  p,\xi\right)  \in
P\times\mathfrak{g}$ into $P\times_{\operatorname*{Ad}}\mathfrak{g}$. It is
not difficult to check that (\ref{monopole 10}) does not depend on the choice
of the fiber point $p\in\pi^{-1}(m)$.

Suppose now that $M$ is a $n$-dimensional Riemann manifold with Riemannian
volume form $\mu\in\Omega^{n}\left(  M\right)  $ and we have a
$\operatorname*{Ad}$-invariant metric $\mathbf{h}$ on $\mathfrak{g}$. For
example, $\mathbf{h}$ could be taken to be (minus) the Killing-Cartan form if
$G$ was a semi-simple compact Lie group. Recall that the inverse of the
Riemann metric on $M$ can be used to define a $C^{\infty}\left(  M\right)
$-bilinear pairing%
\[
\left\langle \cdot,\cdot\right\rangle _{M}:\Omega^{q}\left(  M\right)
\times\Omega^{q}\left(  M\right)  \longrightarrow C^{\infty}\left(  M\right)
,~~~q\in\mathbb{N},
\]
(\cite[Chapter 0]{bleecker}). On the other hand, $\mathbf{h}$ induces a metric
on the fibers of the vector bundle $P\times_{\operatorname*{Ad}}%
\mathfrak{g}\rightarrow M$ in a standard way. We keep on denoting this metric
by $\mathbf{h}$. Both $\left\langle \cdot,\cdot\right\rangle _{M}$ and
$\mathbf{h}$ applied together define a $C^{\infty}\left(  M\right)  $-bilinear
product%
\[
\left\langle \cdot,\cdot\right\rangle :\Omega^{q}\left(  M;\operatorname*{Ad}%
\left(  P\right)  \right)  \times\Omega^{q}\left(  M;\operatorname*{Ad}\left(
P\right)  \right)  \longrightarrow C^{\infty}\left(  M\right)  .
\]
Additionally, the induced metric $\mathbf{h}$ allows us to define a wedge
pairing%
\[
\wedge:\Omega^{r}\left(  M;\operatorname*{Ad}\left(  P\right)  \right)
\times\Omega^{q}\left(  M;\operatorname*{Ad}\left(  P\right)  \right)
\longrightarrow\Omega^{r+q}\left(  M\right)
\]
via the equality%
\[
(\varphi\wedge\psi)(m)\left(  Y_{1},...,Y_{r+q}\right)  :=\frac{1}{r!q!}%
\sum_{\sigma\in S_{r+q}}(-1)^{\left\vert \sigma\right\vert }\mathbf{h}%
_{m}\left(  \varphi\left(  Y_{\sigma(1)},...,Y_{\sigma(r)}\right)
,\varphi\left(  Y_{\sigma(r+1)},...,Y_{\sigma(r+q)}\right)  \right)
\]
for any $\varphi\in\Omega^{r}\left(  M;\operatorname*{Ad}\left(  P\right)
\right)  $, $\psi\in\Omega^{q}\left(  M;\operatorname*{Ad}\left(  P\right)
\right)  $, and any $Y_{1},...,Y_{r+q}\in\mathfrak{X}\left(  M\right)  $. More
importantly, there is a natural operator called the {\bfseries\itshape Hodge
operator}%
\[
\ast:\Omega^{r}\left(  M;\operatorname*{Ad}\left(  P\right)  \right)
\longrightarrow\Omega^{n-r}\left(  M;\operatorname*{Ad}\left(  P\right)
\right)
\]
characterized by the relation%
\[
\theta\wedge\ast\varphi=\left\langle \theta,\varphi\right\rangle \mu\in
\Omega^{n}\left(  M\right)
\]
for any $\varphi\in\Omega^{r}\left(  M;\operatorname*{Ad}\left(  P\right)
\right)  $ and any $\theta\in\Omega^{n-r}\left(  M;\operatorname*{Ad}\left(
P\right)  \right)  $. The Hodge operator defines the inner product%
\[
\left(  \theta,\varphi\right)  :=\int_{M}\theta\wedge\ast\varphi=\int
_{M}\left\langle \theta,\varphi\right\rangle \mu
\]
provided this integral exists. Finally, given $\omega\in\Omega_{equiv}%
^{1}\left(  P;\mathfrak{g}\right)  $, the {\bfseries\itshape covariant
codifferential} $\delta_{\omega}$ is defined by%
\[
\delta_{\omega}\varphi=-(-1)^{n(r+1)}\ast\circ D^{\omega}\circ\ast\varphi
\in\Omega_{equiv}^{r-1}(P;\mathfrak{g})^{\operatorname*{Hor}},~~\varphi
\in\Omega_{equiv}^{r}(P;\mathfrak{g})^{\operatorname*{Hor}},
\]
where we have used the identification $\Omega\left(  M;\operatorname*{Ad}%
\left(  P\right)  \right)  =\Omega_{equiv}(P;\mathfrak{g}%
)^{\operatorname*{Hor}}$ in order to apply the Hodge operator to a
$\mathfrak{g}$-valued horizontal form on $P$.

In a pure Yang-Mills theory, the {\bfseries\itshape Yang-Mills functional}
$YM$ associates to any principal connection $\omega\in\Omega^{1}\left(
P;\mathfrak{g}\right)  $ the real number%
\[
YM\left(  \omega\right)  :=\left(  \Omega^{\omega},\Omega^{\omega}\right)
=\int_{M}\Omega^{\omega}\wedge\ast\Omega^{\omega}.
\]
Roughly speaking, the Yang-Mills functional gives a measure of the total
curvature of the principal connection $\omega$. Critical points of the
functional, the so called {\bfseries\itshape Yang-Mills connections}, are the
most important for physical purposes because their corresponding field
strengths model physical interactions in gauge theories. A classical result
shows that $\omega\in\Omega^{1}\left(  P;\mathfrak{g}\right)  $ is a
Yang-Mills connection if and only if%
\begin{equation}
\delta_{\omega}\Omega^{\omega}=0 \label{monopole 23}%
\end{equation}
(see \cite[Theorem 5.2.3]{bleecker} for a modification of (\ref{monopole 23})
in the presence of currents).

Now, suppose that $\omega\in\Omega^{1}\left(  P;\mathfrak{g}\right)  $ is a
Yang-Mills connection of some bundle $\pi:P\rightarrow\mathbb{S}^{2n}$. We
have already argued that the map (\ref{eq proyeccion homotopica}) can be used
to define principal bundle structures on $\mathbb{R}^{2n+1}\backslash\{0\}$
from those on $\mathbb{S}^{2n}$. Let $F:f^{\ast}\left(  P\right)  \rightarrow
P$ be the bundle homomorphism from the pull-back of $\pi$ by $f:\mathbb{R}%
^{2n+1}\backslash\{0\}\rightarrow\mathbb{S}^{2n}$ given in Equation
(\ref{eq proyeccion homotopica}). The next proposition, whose proof can be
found in the Appendix, shows that the principal connection $F^{\ast}\left(
\omega\right)  \in\Omega^{1}\left(  f^{\ast}\left(  P\right)  ;\mathfrak{g}%
\right)  $ on $f^{\ast}\left(  P\right)  $ is also Yang-Mills.

\begin{proposition}
\label{proposition YM}Let $\pi:P\rightarrow\mathbb{S}^{2n}$ be a principal
bundle with structural group $G$ and let $\omega\in\Omega^{1}\left(
P;\mathfrak{g}\right)  $ be a principal connection. Let $f:\mathbb{R}%
^{2n+1}\backslash\{0\}\rightarrow\mathbb{S}^{2n}$ be as in
(\ref{eq proyeccion homotopica}) and $F^{\ast}:f^{\ast}\left(  P\right)
\rightarrow P$ the corresponding principal bundle homomorphism. Then%
\[
\delta^{F^{\ast}(\omega)}\Omega^{F^{\ast}(\omega)}=-\frac{1}{\overline{\pi
}^{\ast}\left(  r^{2}\right)  }F^{\ast}\left(  \delta^{\omega}\Omega^{\omega
}\right)  ,
\]
where $r\in C^{\infty}\left(  \mathbb{R}^{2n+1}\backslash\{0\}\right)  $ is
the radius function $r\left(  x\right)  =\left\Vert x\right\Vert $,
$x\in\mathbb{R}^{2n+1}\backslash\{0\}$. In particular, $\omega$ is a
Yang-Mills connection if and only if $F^{\ast}\left(  \omega\right)  $ is a
Yang-Mills connection.
\end{proposition}

\section{Principal bundles over homogeneous spaces}

The aim of this section is to introduce the main geometrical ingredients to
study gauge theories over homogeneous spaces. Since we are interested in gauge
theories over the $n$-dimensional sphere $\mathbb{S}^{n}=SO(n+1)/SO(n)$.
However, among all the possible principal bundle structures over
$\mathbb{S}^{n}$, we need to characterize those admitting a (left)
$SO(n+1)$-action in order to talk properly about spherically symmetric
quantities. This will be done in the first subsection. We will see that these
principal bundles can be labelled by a Lie group homomorphism $\lambda
:SO(n)\rightarrow G$ from the isotropy group to the gauge group. Moreover,
they can be understood as homogeneous spaces themselves, a perspective that
will be extremely fruitful. At the end we will give a characterization of the
four more relevant examples of our study, the principal bundles which will
correspond to Dirac, Yang, and $SO(2k)$-monopoles, $k\in\mathbb{N}$. Once we
have learnt how to build such principal bundles, we will characterize in
Subsection \ref{subsection invariant PC} the principal connections (gauge
potentials) which are invariant by the rotations group in terms of linear maps
$W:\mathfrak{so}(n)\rightarrow\mathfrak{g}$ satisfying some compatibility
conditions. Finally, in Subsection \ref{subsection symmetric spaces}, we
introduce symmetric spaces, a particular subclass of homogeneous spaces whose
Lie algebra can be suitably decomposed. For example, the sphere $\mathbb{S}%
^{n}$ is a symmetric space. Over them, we will show that there exists a unique
$SO(n+1)$-invariant principal connection; that is, a monopole potential
one-form. This means that requiring a principal bundle to admit a
$SO(n+1)$-action equals to having an essentially unique spherically symmetric
configuration on it.

\subsection{Homogeneous principal bundles}

Let $K$ and $G$ be two Lie groups and $H\subset K$ a closed subgroup. A
{\bfseries\itshape homogeneous principal bundle} $\pi:P\rightarrow K/H$ with
structural group $G$ is a principal bundle over a homogeneous space $K/H$
together with a left $K$-action on $P$ by automorphisms which projects to the
left multiplication of $K$ on the base manifold $K/H$. According to
\cite{harnard} and \cite{wang}, homogeneous principal bundles $\pi
:P\rightarrow K/H$ with structural group $G$ are (modulo isomorphisms) in
one-to-one correspondence with group homomorphisms $\lambda:H\rightarrow G$
(modulo conjugation) so that $\pi:P\rightarrow K/H$ is isomorphic to the
{\bfseries\itshape associated bundle} $P_{\lambda}:=K\times_{H}G$; that is,
the space of orbits of the right action%
\begin{equation}%
\begin{array}
[c]{rrl}%
\Psi_{\lambda}:\left(  K\times G\right)  \times H & \longrightarrow & K\times
G\\
\left(  \left(  k,g\right)  ,h\right)  & \longmapsto & \left(  kh,\lambda
(h)^{-1}g\right)  .
\end{array}
\label{monopole 4}%
\end{equation}
Denoting the elements $p$ of $P_{\lambda}$ as equivalent classes,
$p=[k,g]^{\sim}$ such that $k\in K$ and $g\in G$, the projection $\pi$ is
simply given by $[k,g]^{\sim}\longmapsto kH\in K/H$. If $p\in\pi^{-1}\left(
o\right)  $ is some point in the equivalence class $o\in K/H$ of $e\in K$, the
homomorphism $\lambda:H\rightarrow G$ can be understood by the relation%
\[
h\cdot p=p\cdot\lambda(h),~~h\in H,
\]
where the dot $\cdot$ denotes the left action of $K$ or the right action of
$G$ on $P_{\lambda}$ respectively. We encourage the reader to check with
\cite{michor} for a brief review on the basic facts about associated bundles.

Furthermore, $P_{\lambda}$ can be also seen as the {homogeneous space}
$\left.  \left(  K\times G\right)  \right/  \widetilde{H}$, where
$\widetilde{H}$ is the closed subgroup $\widetilde{H}=\{\left(  h,\lambda
(h)\right)  ~|~h\in H\}\subset K\times G$, clearly isomorphic to $H$: that is
why the principal bundles $P_{\lambda}$ are called homogeneous. The
isomorphism works as follows:%
\begin{equation}%
\begin{array}
[c]{rrl}%
\Upsilon:\left.  \left(  K\times G\right)  \right/  \widetilde{H} &
\longrightarrow & P_{\lambda}\\
\overline{(k,g)} & \longmapsto & [k,g^{-1}]^{\sim},
\end{array}
\label{monopole 8}%
\end{equation}
where $\overline{(k,g)}$ and $[k,g^{-1}]^{\sim}$ denote the equivalent class
of $(k,g)\in K\times G$ in $\left.  \left(  K\times G\right)  \right/
\widetilde{H}$ and $P_{\lambda}$ respectively.

Finally, we fix some notation for later convenience. The \textit{left} action
$L_{P_{\lambda}}:K\times P_{\lambda}\rightarrow P_{\lambda}$ and a
\textit{right} action $R_{\lambda}:G\times P_{\lambda}\rightarrow P_{\lambda}$
that we have on a homogeneous principal bundle are respectively given by%
\begin{equation}
(L_{\lambda})_{\bar{k}}\left(  [k,g]^{\sim}\right)  =[\bar{k}k,g]^{\sim}\text{
and }(R_{\lambda})_{\bar{g}}\left(  [k,g]^{\sim}\right)  =[k,g\bar{g}]^{\sim
}\text{, \ }g,\bar{g}\in G\text{, }k,\bar{k}\in K. \label{monopole 2}%
\end{equation}

\begin{remark}
\normalfont In the general classification theory of bundles, two principal
bundles with the same base manifold and the same structural group are called
equivalent if there exists a homomorphism between them which projects onto the
identity map on the basis. When the base manifold is the $n$-dimensional
sphere $\mathbb{S}^{n}$, such equivalence classes are in bijection with the
elements of the homotopy group $\pi_{n-1}(G)$ provided the gauge group $G$ is
connected (see \cite{steenrod}). Take for example $n=3$ and $G=SO(3)$. Since
$\pi_{2}\left(  SO(3)\right)  =0$, we know that, essentially, there exists a
unique principal bundle over $\mathbb{S}^{3}$ with structural group $SO(3)$.
Namely, $\pi:SO(4)\rightarrow\mathbb{S}^{3}=SO(4)/SO(3)$. Therefore,
$\pi:SO(4)\rightarrow\mathbb{S}^{3}$ is trivializable and $SO(4)$ is
diffeomorphic to $SO(3)\times\mathbb{S}^{3}$. However, they are not isomorphic
as Lie groups (see Proposition \ref{prop 2}). On the other hand, there exist
at least two homomorphisms $\lambda:SO(3)\rightarrow G=SO(3)$ which are not
conjugated: the trivial homomorphism $\lambda\left(  h\right)  =e\in SO(3)$
for any $h\in SO(3)$, and the identity homomorphism, $\lambda
=\operatorname*{Id}$. So, according to what we have said so far, there exist
two different principal bundles over $\mathbb{S}^{3}$ with gauge group $SO(3)$
admitting a left action of $SO(4)$. Is this a contradiction? The answer is no.
Everything relies on the notion of \textit{equivalence} of principal bundles
we use. In general, when we forget about the $SO(4)$-left action, there always
exists a fiber preserving diffeomorphism between two any principal bundles
over $\mathbb{S}^{3}$. But my notion of equivalence changes when $SO(4)$ acts
upon our principal bundles in the way we stated. Then, the previous
diffeomorphism needs to be also equivariant with respect to the two $SO(4)$
actions, a requirement that prevents some bundles from being equivalent. In
other words, we can define at least two different $SO(4)$-left actions on the
unique principal bundle over $\mathbb{S}^{3}$ with gauge group $SO(3)$ in a
non-equivalent way.
\end{remark}

\begin{remark}
\normalfont The theory of {\bfseries\itshape equivariant principal bundles}
tries to describe those principal bundles $\pi:P\rightarrow M$ with structural
Lie group $G$ such that both $P$ and $M$ are left acted upon another Lie group
$K$ such that the projection $\pi$ is $K$-equivariant and the actions of $K$
and $G$ commute. This is a much more general framework that reduces to ours
when $M=K/H$ is a homogeneous space, where $H\subseteq K$ is a closed Lie
subgroup. Under some general assumptions and in particular for the case
$\mathbb{S}^{n}=SO(n+1)/SO(n)$, $n\geq3$, it can be checked that the number of
isomorphic principal bundles $\pi:P\rightarrow M$ with structural group $G$
over a left $K$-manifold $M$ is finite provided that $G$ is compact and the
isotropy groups $K_{m}$ are semi-simple, $m\in M$ (\cite[Corollary
8.6]{hambelton}). In particular, the number of principal bundles over
$\mathbb{S}^{n}$, $n\geq3$, with structural group $G$ compact admitting a
$SO(n+1)$-left action is finite.
\end{remark}

\begin{examples}
\normalfont\label{examples 1}Let $\mathcal{R}(n,G)$ be the set of smooth
homomorphisms from $SO(n)$ to $G$ modulo conjugation by elements of $G$. We
will describe $\mathcal{R}(n,G)$ for some values of $n\in\mathbb{N}$ and some
Lie groups $G$ that will allow us to study later on some of the monopole
configurations found in the literature (see \cite{hambelton} and references therein).

\begin{enumerate}
\item[\textbf{(i)}] $n=2$ and $G=U(1)$. Given that $SO(2)=U(1)$, the set of
homomorphisms $\mathcal{R}(2,U(1))$ is $\lambda:U(1)\rightarrow U(1)$ modulo
conjugation. It is well known that such a set can be labelled by $\mathbb{Z}$,
the set of integers. Regarding $U(1)=\{\operatorname*{e}\nolimits^{iz}%
:z\in\lbrack0,2\pi)\}$, we can chose the homomorphisms%
\[%
\begin{array}
[c]{rrl}%
\lambda_{m}:U(1) & \longrightarrow & U(1)\\
\operatorname*{e}\nolimits^{iz} & \longmapsto & \left(  \operatorname*{e}%
\nolimits^{iz}\right)  ^{m}=\operatorname*{e}\nolimits^{izm},
\end{array}
~~m\in\mathbb{Z},
\]
as representatives of the equivalent classes of $\mathcal{R}(2,U(1))$.

\item[\textbf{(ii)}] $n=4$ and $G=SO(3)$. The {\bfseries\itshape algebra of
quaternions} $\mathbb{H}$ is usually defined abstractly as a $4$-dimensional
real vector space with a multiplication $\left(  x,y\right)  \mapsto xy$,
$x,y\in\mathbb{H}$, which satisfies the usual associative and distributive
laws and with a distinguished basis $\{\mathbf{1,i,j,k}\}$ satisfying the
following commutation relations%
\begin{align*}
\mathbf{i}^{2}  &  =\mathbf{j}^{2}=\mathbf{k}^{2}=-\mathbf{1}\\
\mathbf{ij}  &  =-\mathbf{ji}=\mathbf{k},~\mathbf{jk}=-\mathbf{kj}%
=\mathbf{i},~\mathbf{ki}=-\mathbf{ik}=\mathbf{j}.
\end{align*}
The {\bfseries\itshape modulus} of a quaternion $x=x_{0}\mathbf{1}%
+x_{1}\mathbf{i}+x_{2}\mathbf{j}+x_{3}\mathbf{k}$ is $\left\vert x\right\vert
=\left(  x_{0}^{2}+x_{1}^{2}+x_{2}^{2}+x_{3}^{2}\right)  ^{1/2}$. The set of
unit quaternions $S^{3}:=\{x\in\mathbb{H}~|~\left\vert x\right\vert =1\}$ is
isomorphic to $SU(2)$ and homeomorphic to the $3$-sphere $\mathbb{S}%
^{3}\subset\mathbb{R}^{4}$ (\cite[Theorem 1.1.4]{naber foundations}).
Moreover, $S^{3}\times S^{3}$ is the universal covering group of $SO(4)$
(\cite[Example 4.32]{michor}) so that $SO(4)\cong\left(  S^{3}\times
S^{3}\right)  /\{(\mathbf{1},\mathbf{1}),\left(  -\mathbf{1},-\mathbf{1}%
\right)  \}$. On the other hand, $S^{3}=SU(2)$ is the universal covering group
of $SO(3)\cong S^{3}/\{\pm\mathbf{1}\}$. The set $\mathcal{R}(4,SO(3))$
contains three elements: the trivial homomorphism and those induced from the
projections $S^{3}\times S^{3}\rightarrow S^{3}$ given by $\sigma_{1}\left(
x,y\right)  =x$ and $\sigma_{2}\left(  x,y\right)  =y$.

\item[\textbf{(iii)}] $n=4$ and $G=SO(4)$. Using the identification
$SO(4)\cong\left(  \mathbb{S}^{3}\times\mathbb{S}^{3}\right)  /\{(\mathbf{1}%
,\mathbf{1}),\left(  -\mathbf{1},-\mathbf{1}\right)  \}$ as in \textbf{(ii)},
the set $\mathcal{R}(4,SO(3))$ contains five elements: the trivial
homomorphism, the identity $\operatorname*{Id}:SO(4)\rightarrow SO(4)$, which
give rise to the principal bundle $SO(5)\rightarrow SO(5)/SO(4)$, and three
homomorphisms induced by the maps $\sigma_{3},\sigma_{4},\delta:\mathbb{S}%
^{3}\times\mathbb{S}^{3}\rightarrow\mathbb{S}^{3}\times\mathbb{S}^{3}$ given
by $\sigma_{3}(x,y)=\left(  x,x\right)  $, $\sigma_{4}\left(  x,y\right)
=\left(  y,y\right)  $, and $\delta\left(  x,y\right)  =\left(  y,x\right)  $.

\item[\textbf{(iv)}] $n=2k\geq6$ and $G=SO(2k)$, $k\in\mathbb{N}$. The set
$\mathcal{R}(2k,SO(2k))$ contains three elements: the trivial homomorphism,
the identity $\operatorname*{Id}:SO(2k)\rightarrow SO(2k)$, whose associated
principal bundle is $SO(2k+1)\rightarrow SO(2k+1)/SO(2k)$, and the conjugation
$\delta$ by the diagonal matrix $\left(  -1,...,-1,1\right)  $. Observe that
$\delta\in O(2k)$ but $\delta\notin SO(2k)$. \ \ \ \ $\blacksquare$
\end{enumerate}
\end{examples}

\subsection{Invariant principal connections\label{subsection invariant PC}}

Let $\pi:P\rightarrow K/H$ be a homogeneous principal bundle as in the
previous subsection. We say that a principal connection $\omega$ is
{\bfseries\itshape$K$-invariant} if $(L_{\lambda})_{k}^{\ast}\omega=\omega$
for any $k\in K$. One can prove that, if $\mathfrak{k}$, $\mathfrak{h}$, and
$\mathfrak{g}$ denote the Lie algebra of $K$, $H$, and the gauge group $G$
respectively, $K$-invariant principal connections on $\pi:P_{\lambda
}\rightarrow K/H$ are in one-to-one correspondence with linear maps
$W:\mathfrak{k}\rightarrow\mathfrak{g}$ such that

\begin{enumerate}
\item[\textbf{(i)}] $W(\xi)=T_{e}\lambda\left(  \xi\right)  $ for any $\xi
\in\mathfrak{h}$,

\item[\textbf{(ii)}] $W\left(  \operatorname*{Ad}\nolimits_{h}\xi\right)
=\operatorname*{Ad}\nolimits_{\lambda(h)}\left(  W(\xi)\right)  $ for any
$\xi\in\mathfrak{k}$ and any $h\in H$.
\end{enumerate}

\noindent(see \cite{wang}, \cite{kobayasi}). From now on, we are going to
refer to these linear maps $W$ as {\bfseries\itshape Wang maps}. Given a Wang
map $W:\mathfrak{k}\rightarrow\mathfrak{g}$, the principal connection
$\omega\in\Omega^{1}(P_{\lambda};\mathfrak{g})$ is given by%
\begin{equation}
\omega_{p_{o}}(\xi_{P_{\lambda}})=W(\xi) \label{monopole 1}%
\end{equation}
where $\xi\in\mathfrak{k}$, $o$ denotes the equivalent class of $e\in K$ in
$K/H$, $p_{o}\in\pi^{-1}(o)$ is any arbitrary point on the fiber of $o\in
K/H$, and $\xi_{P_{\lambda}}$ is the vector field induced on $P_{\lambda}$ by
the $K$-action (\cite[Theorem 11.5]{kobayasi}). Observe that, since $\omega$
is $K$-invariant and $G$ acts transitively on the fibers of $P_{\lambda}$,
(\ref{monopole 1}) and (\ref{monopolo 18 a}) characterizes $\omega$ completely.

One of the most important examples of homogeneous spaces are those called
reductive. Recall that a homogeneous space $K/H$ is called
{\bfseries\itshape
reductive} if the Lie algebra $\mathfrak{k}$ can be written as $\mathfrak{k}%
=\mathfrak{h}\oplus\mathfrak{m}$ and $\operatorname*{Ad}\nolimits_{h}\left(
\mathfrak{m}\right)  \subseteq\mathfrak{m}$. For a reductive homogeneous space
$K/H$, the linear map $\mathbf{W}:\mathfrak{k}\rightarrow\mathfrak{g}$ defined
as $\left.  \mathbf{W}\right\vert _{\mathfrak{h}}=T_{e}\lambda$ and $\left.
\mathbf{W}\right\vert _{\mathfrak{m}}=0$ is called the {\bfseries\itshape
canonical connection}.

\begin{example}
\label{example maurer-cartan}\normalfont It can be shown that the principal
$H$-bundle $K\rightarrow K/H$ admits a $K$-invariant connection if and only if
$K/H$ is reductive (\cite[Theorem 11.1]{kobayasi}). The canonical connection
$\boldsymbol{\omega}\in\Omega^{1}\left(  K,\mathfrak{h}\right)  $ on
$K\rightarrow K/H$ is given by the $\mathfrak{h}$-valued part of the
{\bfseries\itshape Maurer-Cartan form} $\omega_{MC}$ which is defined by
$\omega_{MC}(k)\left(  \xi_{K}(k)\right)  =\xi\in\mathfrak{k}$, $k\in K$. That
is, $\boldsymbol{\omega}(k)(\xi_{K}(k))=\operatorname*{proj}_{\mathfrak{h}%
}(\xi)$. \ \ \ \ $\blacksquare$
\end{example}

Principal connections can be used to induce connections on associated bundles
(see \cite[19.8]{michor} and subsequent sections for a general approach to
this subject). The details of this mechanism and the proof of the following
proposition, that we include here for the sake of a more complete exposition,
are postponed to the Appendix \ref{appendix prop 1}. The proposition claims
that the principal connection of $P_{\lambda}$ is induced from that of
$K\rightarrow K/H$ whenever $K/H$ is reductive. Although it can be found in
the literature, its proof is frequently omitted, so we decided to prove it
ourselves explicitly.

\begin{proposition}
\label{prop 1}Let $K/H$ be reductive. Then, the canonical connection on
$P_{\lambda}$ is induced from the canonical connection of $K\rightarrow K/H$.
\end{proposition}

\subsection{Symmetric spaces\label{subsection symmetric spaces}}

We are now going to describe invariant connections over a particular class of
homogeneous spaces: symmetric spaces. Symmetric spaces are usually presented
in the context of Riemannian geometry. Most of the content of this subsection
is extracted from \cite{kobayashi 2}, which the reader is encourage to check
with. We will see that, over a symmetric space $K/H$, the canonical connection
is the unique principal connection which is $K$-invariant. Since the sphere
$\mathbb{S}^{n}$ is a symmetric space, it means that there will exist a unique
monopole configuration on any homogeneous principal bundle over $\mathbb{S}%
^{n}$.

Let $M$ be a $n$-dimensional Riemann manifold with an
{\bfseries\itshape affine connection} $\nabla$, that is, a connection in the
frame bundle. Let $U\subseteq M$ be an open neighborhood, $x\in U$ a fixed
point, and $X_{x}\in T_{x}M$. Denote by $\exp\left(  X_{x}\right)  $ the value
of the geodesic $\gamma\left(  t\right)  $ at time $t=1$ which satisfies
$\gamma\left(  0\right)  =x$, $\dot{\gamma}\left(  0\right)  =X_{x}$. This
value exists for $X_{x}$ in a suitable small neighborhood of $0\in T_{x}M$. A
diffeomorphism $\varphi:M\rightarrow M$ is called an {\bfseries\itshape affine
transformation} if it is a diffeomorphism and $T\varphi:TM\rightarrow TM$ maps
each parallel vector field along a curve $\tau:(-\varepsilon,\varepsilon
)\rightarrow M$, $\varepsilon>0$, into a parallel vector field along the curve
$\varphi(\tau)$. A {\bfseries\itshape symmetry $s_{x}$} at a point $x\in U$ is
a diffeomorphism of $U$ onto itself which sends $\exp\left(  X_{x}\right)  $
into $\exp\left(  -X_{x}\right)  $. Observe that a symmetry $s_{x}$ is
involutive: $s_{x}\circ s_{x}=\operatorname*{Id}$. If there exists an affine
transformation $s_{x}$ for any $x\in M$, then $M$ is said to be
{\bfseries\itshape affine locally symmetric}. $M$ is said {\bfseries\itshape
affine symmetric} if the symmetry $s_{x}$ can be extended to a global affine
transformation of $M$ for any $x\in M$.

The group of affine transformations of an affine symmetric manifold $M$ is a
Lie group which acts transitively on it (\cite[Chapter XI, Theorem
1.4]{kobayashi 2}). If $K$ denotes the identity component of such group, then
$M=K/H$, where $H$ denotes the subgroup of those affine transformations in $K$
leaving a point $o\in M$ fixed (\cite[Chapter IV, Theorem 3.3]{helgason}).
Taking this remark into account, we say that a triple $\left(  K,H,\sigma
\right)  $ is a {\bfseries\itshape symmetric space} if $K,H$ are Lie groups,
$H\subset K$, $\sigma:K\rightarrow K$ is an involutive automorphism, and
$K_{\sigma}^{e}\subseteq H\subseteq K_{\sigma}$. Here $K_{\sigma}$ denotes the
set of elements of $K$ which are invariant by $\sigma$ and $K_{\sigma}^{e}$
the identity component of $K_{\sigma}$. In the case of an affine symmetric
manifold $M$, the automorphism $\sigma$ is given by $\sigma\left(  k\right)
=s_{o}\circ k\circ s_{o}^{-1}$ where $s_{o}$ is a symmetry at $o$. On the
contrary, each symmetry $s_{x}$ can be recovered from $\sigma$ as
$s_{x}=k\circ s_{o}\circ k^{-1}$, $x\in M$. In general, $s_{o}$ is defined to
be the involutive diffeomorphism of $K/H$ onto itself induced by the
automorphism $\sigma$.

\begin{example}
\normalfont The $n$-dimensional sphere $\mathbb{S}^{n}$ is a symmetric space.
Indeed, if $K=SO\left(  n+1\right)  $ and $o=(1,0,\overset{n-2)}{...}%
,0)\in\mathbb{S}^{n}\subset\mathbb{R}^{n+1}$, then%
\[
H=\left(
\begin{array}
[c]{cc}%
1 & 0\\
0 & SO(n)
\end{array}
\right)  \cong SO(n)
\]
and $\mathbb{S}^{n}=SO(n+1)/SO(n)$. \ \ \ \ $\blacksquare$
\end{example}

In terms of the Lie algebras $\mathfrak{k}$ and $\mathfrak{h}$ of $K$ and $H$,
respectively, a symmetric space $\left(  K,H,\sigma\right)  $ is described as
follows. To start with, we see from the involutivity of $\sigma$ that
$T_{e}\sigma:\mathfrak{k}\rightarrow\mathfrak{k}$ has eigenvalues $+1$ and
$-1$. Then, the Lie algebra $\mathfrak{k}$ can be written as $\mathfrak{h}%
\oplus\mathfrak{m}$, where $\mathfrak{h}$ is the eigenspace associated to the
eigenvalue $1$ and $\mathfrak{m}$ is the eigenspace associated to $-1$.
Moreover,%
\[
\lbrack\mathfrak{h},\mathfrak{h}]\subset\mathfrak{h,~\ }[\mathfrak{h}%
,\mathfrak{m}]\subset\mathfrak{m,~\ }[\mathfrak{m},\mathfrak{m}]\subset
\mathfrak{h}%
\]
and $\operatorname*{Ad}_{H}(\mathfrak{m})\subset\mathfrak{m}$ (\cite[Chapter
XI, Proposition 2.1 and 2.2]{kobayashi 2}). That is, symmetric spaces are reductive.

When the gauge group $G$ is a subgroup of $GL\left(  n;\mathbb{R}\right)  $,
homogeneous principal bundles $P_{\lambda}\rightarrow K/H$ can be regarded as
subbundles of the frame bundle. This is the case in our examples. Then, any
$K$-invariant principal connection on $P_{\lambda}$ (i.e., a Wang map
$W:\mathfrak{k}\rightarrow\mathfrak{g\subset gl}\left(  n;\mathbb{R}\right)
$) can be consequently considered as $K$-invariant affine connection (i.e., a
Wang map $W:\mathfrak{k}\rightarrow\mathfrak{gl}\left(  n;\mathbb{R}\right)
$). The next theorem is the most important as far as characterizing invariant
affine connections on symmetric spaces is concerned.

\begin{theorem}
[{\cite[Theorem 3.1 and 3.3]{kobayashi 2}}]Let $\left(  K,H,\sigma\right)  $
be a symmetric space. The canonical connection is the only affine connection
on $K/H$ which is invariant by the symmetries $s_{x}$ of $M$, $x\in M$.
Furthermore, a $K$-invariant (indefinite) Riemannian metric on $K/H$, if there
exists any, induces the canonical connection on $M$.
\end{theorem}

The previous theorem is important for the following reason. We defined
monopoles as those configurations invariant by $SO(n)$ because elements of
$SO(n)$ are physically relevant symmetries of our base space-time
$\mathbb{R}^{n}\backslash\{0\}$. However, in more general models, there may
not exist any natural action of $SO(n)$ onto the base manifold $M$, which is
supposed to be a Riemann manifold according to General Relativity. In this
case, the group of symmetries $s_{x}$ seems to be the natural candidate to
replace $SO(n)$ in the definition of spherical symmetry. In other words, we
should require monopoles to be invariant by the symmetries $s_{x}$, $x\in M$,
instead of by $SO(n)$.

Nevertheless, we are interested so far in connections which are invariant not
by the symmetries but by the action of $K$. As Laquer shows in \cite{laquer},
except for very concrete cases, the canonical connection is the unique affine
connection on a symmetric space $\left(  K,H,\sigma\right)  $ which is
$K$-invariant. Therefore, the unique connection available to construct monopoles.

\begin{theorem}
[{\cite[Theorem 2.1]{laquer}}]\label{teorema laquer}Let $K$ be a simple Lie
group and $\left(  K,H,\sigma\right)  $ a symmetric space. The set of
$K$-invariant affine connections on $K/H$ consists of just the canonical
connection in all cases except for the following:%
\begin{equation}%
\begin{array}
[c]{rr}%
SU\left(  n\right)  /SO(n) & ~~n\geq3,\\
SU(2n)/SP(n) & ~~n\geq3,\\
E_{6}/F_{4}. &
\end{array}
\label{monopole 14}%
\end{equation}
Each of these spaces has a one-dimensional family of invariant affine connections.
\end{theorem}

\section{The algebraic setting\label{section algebraic setting}}

In this section we are going to describe algebraically the space
$\Omega_{equiv}\left(  P_{\lambda};\mathfrak{g}\right)  ^{K}$ of
$\mathfrak{g}$-valued forms which are $G$-equivariant in the sense of
(\ref{monopolo 18 a}) and $K$-invariant by the left action $L_{\lambda}$
(\ref{monopole 2}). The field strength $\Omega^{\omega}$ will be then a
multilinear map from $\mathfrak{k}$ to $\mathfrak{g}$ easily expressed in
terms of the corresponding Wang map. Carrying out such identification is quite
simple. Since two arbitrary points in $P_{\lambda}$ are always linked by the
composition of the actions of $K$ and $G$ on $P_{\lambda}$, any $\alpha
\in\Omega_{equiv}\left(  P_{\lambda};\mathfrak{g}\right)  ^{K}$ is fully
characterized by its values on a fixed point $p\in P_{\lambda}$. Suppose that
$p\in\pi^{-1}(o)$ is $p=[e,e]^{\sim}$ as in the proof of Proposition
\ref{prop 1}. Since the isomorphism (\ref{monopole 8}) allows us to identify
$T_{p}P_{\lambda}$ with $(\mathfrak{k}\times\mathfrak{g})/\widetilde
{\mathfrak{h}}$, it seems reasonable to express $\Omega_{equiv}\left(
P_{\lambda};\mathfrak{g}\right)  ^{K}$ as a suitable set of forms defined on
$\mathfrak{k}\times\mathfrak{g}$ satisfying some restrictions. We will
particularize in Subsection \ref{examples field strength} the canonical field
strengths of the homogeneous principal bundles introduced in Examples
\ref{examples 1}, which will correspond to the field strengths of Dirac, Yang,
and $SO(2n)$-monopoles, $n\in\mathbb{N}$. Moreover, we will also prove that
they satisfy the Yang-Mills connections (Proposition
\ref{prop canonical is YM}) and, therefore, give rise to monopole
configurations indeed.

First of all, observe that $\Omega_{equiv}\left(  P_{\lambda};\mathfrak{g}%
\right)  ^{K}$ coincides with the space of $\mathfrak{g}$-valued forms forms
on $P_{\lambda}$ $\left(  K\times G\right)  $-equivariant with respect to the
left $\left(  K\times G\right)  $-actions%
\begin{equation}%
\begin{array}
[c]{rrl}%
\Psi:\left(  K\times G\right)  \times P_{\lambda} & \longrightarrow &
P_{\lambda}\\
\left(  \left(  k,g\right)  ,[k_{2},g_{2}]^{\sim}\right)  & \longmapsto &
[kk_{2},g_{2}g^{-1}]^{\sim}%
\end{array}
\label{monopole 19}%
\end{equation}
and%
\begin{equation}%
\begin{array}
[c]{rrl}%
\rho:\left(  K\times G\right)  \times\mathfrak{g} & \longrightarrow &
\mathfrak{g}\\
\left(  \left(  k,g\right)  ,\xi\right)  & \longmapsto & \operatorname*{Ad}%
_{g}\xi.
\end{array}
\label{monopole 20}%
\end{equation}
That is, $\varphi\in\Omega_{equiv}\left(  P_{\lambda};\mathfrak{g}\right)
^{K}$ if and only if $\Psi_{(k,g)}^{\ast}\left(  \varphi\right)  =\rho
_{(k,g)}\circ\varphi$ for any $\left(  k,g\right)  \in K\times G$. On the
other hand, if
\begin{equation}%
\begin{array}
[c]{rrl}%
\widetilde{\Psi}:\left(  K\times G\right)  \times\left.  (K\times G)\right/
\widetilde{H} & \longrightarrow & \left.  (K\times G)\right/  \widetilde{H}\\
\left(  \left(  k,g\right)  ,\overline{\left(  k_{2},g_{2}\right)  }\right)  &
\longmapsto & \overline{\left(  kk_{2},gg_{2}\right)  }%
\end{array}
\label{monopole 25}%
\end{equation}
is the natural left action of $K\times G$ on the quotient space $\left.
(K\times G)\right/  \widetilde{H}$, the isomorphism $\Upsilon:\left.  (K\times
G)\right/  \widetilde{H}\rightarrow P_{\lambda}$ introduced in
(\ref{monopole 8}) is such that the following diagram commutes%
\[%
\begin{array}
[c]{rcl}%
P_{\lambda} & \overset{\Psi_{\left(  k,g\right)  }}{\longrightarrow} &
P_{\lambda}\\
{\small \Upsilon}\uparrow & \# & \uparrow~{\small \Upsilon}\\
\left.  (K\times G)\right/  \widetilde{H} & \underset{\widetilde{\Psi
}_{\left(  k,g\right)  }}{\longrightarrow} & \left.  (K\times G)\right/
\widetilde{H}%
\end{array}
\]
for any $\left(  k,g\right)  \in K\times G$. Therefore, $\varphi\in
\Omega_{equiv}\left(  P_{\lambda};\mathfrak{g}\right)  ^{K}$ if and only if%
\[
\widetilde{\Psi}_{\left(  k,g\right)  }^{\ast}\left(  \Upsilon^{\ast}%
\varphi\right)  =\rho_{\left(  k,g\right)  }\circ\Upsilon^{\ast}%
(\varphi)=\operatorname*{Ad}\nolimits_{g}\circ\Upsilon^{\ast}(\varphi).
\]
In this situation, provided that $K\times G$ is connected, one of the
consequences of \cite[Theorem 13.1]{chevalley cohomology} is that the space
$\Omega_{equiv}(\left.  (K\times G)\right/  \widetilde{H};\mathfrak{g)}$\ of
$\mathfrak{g}$-valued forms on $\left.  (K\times G)\right/  \widetilde{H}$
which are $(K\times G)$-invariant with respect to the actions
(\ref{monopole 25}) and (\ref{monopole 20}) is isomorphic to the graded
differential algebra $\Lambda_{\widetilde{\mathfrak{h}}}\left(  \mathfrak{k}%
\times\mathfrak{g};\mathfrak{g}\right)  $, the space of $\mathfrak{g}$-valued
chains on $\mathfrak{k}\times\mathfrak{g}$ such that

\begin{enumerate}
\item[\textbf{(i)}] vanish on $\widetilde{\mathfrak{h}}=\{\xi\in
\mathfrak{h}~|~(\xi,T_{e}\lambda(\xi))\in\mathfrak{k}\times\mathfrak{g}\}$ and

\item[\textbf{(ii)}] if $\varphi\in\Lambda^{n}\left(  \mathfrak{k}%
\times\mathfrak{g};\mathfrak{g}\right)  $, $z,z_{1},...,z_{n}\in
\mathfrak{k}\times\mathfrak{g}$, $z=\left(  \xi,\eta\right)  $, $z_{i}=\left(
\xi_{i},\eta_{i}\right)  $ with $\xi,\xi_{i}\in\mathfrak{k}$ and $\eta
,\eta_{i}\in\mathfrak{g}$ for any $i=1,...,n$, then%
\begin{equation}
\left[  T_{e}\lambda(\xi),\varphi\left(  z_{1},...,z_{n}\right)  \right]
=\sum_{i=1}^{n}\varphi\left(  z_{1},...,[z,z_{i}],...,z_{n}\right)  ,
\label{monopole 9}%
\end{equation}
where $[z,z_{i}]=\left(  \left[  \xi,\xi_{i}\right]  ,\left[  T_{e}\lambda
(\xi),\eta_{i}\right]  \right)  \in\mathfrak{k}\times\mathfrak{g}$.
\end{enumerate}

Let%
\[
\Phi:\Omega_{equiv}(\left.  (K\times G)\right/  \widetilde{H};\mathfrak{g)}%
\cong\Lambda_{\widetilde{\mathfrak{h}}}\left(  \mathfrak{k}\times
\mathfrak{g};\mathfrak{g}\right)
\]
be the isomorphism between $\Omega_{equiv}(\left.  (K\times G)\right/
\widetilde{H};\mathfrak{g)}$ and $\Lambda_{\widetilde{\mathfrak{h}}}\left(
\mathfrak{k}\times\mathfrak{g};\mathfrak{g}\right)  $. For example, $\Phi$
sends a principal connection $\omega\in\Omega_{equiv}^{1}\left(  P_{\lambda
};\mathfrak{g}\right)  $ associated to a Wang map $W:\mathfrak{k}%
\rightarrow\mathfrak{g}$ to the one chain $\widetilde{W}:\mathfrak{k}%
\times\mathfrak{g}\rightarrow\mathfrak{g}$ given by $\widetilde{W}(\xi
,\eta)=W(\xi)-\eta$, $\xi\in\mathfrak{k}$, $\eta\in\mathfrak{g}$. We define
the {\bfseries\itshape horizontal projector} $\operatorname*{Hor}%
\nolimits_{\widetilde{W}}:\mathfrak{k}\times\mathfrak{g}\rightarrow
\mathfrak{k}\times\mathfrak{g}$ as $\operatorname*{Hor}\nolimits_{\widetilde
{W}}(\xi,\eta)=(\xi,W(\xi))$ and the {\bfseries\itshape vertical projector}
$\operatorname*{Ver}\nolimits_{\widetilde{W}}:\mathfrak{k}\times
\mathfrak{g}\rightarrow\mathfrak{k}\times\mathfrak{g}$ $\operatorname*{Ver}%
\nolimits_{\widetilde{W}}(\xi,\eta)=(0,\eta-W(\xi))$. We made the dependence
on the Wang map $W$ explicit in order to distinguish these vertical and
horizontal projectors from those associated to $TP_{\lambda}$ and $\omega$. In
this context, the {\bfseries\itshape exterior differential} operator
$\mathbf{d}:\Lambda^{n}\left(  \mathfrak{k}\times\mathfrak{g};\mathfrak{g}%
\right)  \rightarrow\Lambda^{n+1}\left(  \mathfrak{k}\times\mathfrak{g}%
;\mathfrak{g}\right)  $ is defined as
\begin{align*}
\mathbf{d}\varphi\left(  z_{1},...,z_{n+1}\right)   &  =\sum_{i=1}%
^{n+1}(-1)^{i-1}[\eta_{i},\varphi\left(  z_{1},...,\widehat{z}_{i}%
,...,z_{n+1}\right)  ]\\
&  +\sum_{i<j}\left(  -1\right)  ^{i+j}\varphi\left(  \left(  \lbrack\xi
_{i},\xi_{j}],[\eta_{i},\eta_{j}]\right)  ,z_{1},...,\widehat{z}%
_{i},...,\widehat{z}_{j},...,z_{n+1}\right)  ,
\end{align*}
where $z_{i}=\left(  \xi_{i},\eta_{i}\right)  $ with $\xi_{i}\in\mathfrak{k}$
and $\eta_{i}\in\mathfrak{g}$ for any $i=1,...,n+1$. In the same way that we
introduced the covariant derivative $D^{\omega}$ on $\Omega\left(  P_{\lambda
};\mathfrak{g}\right)  $ from a principal connection $\omega\in\Omega
^{1}\left(  P_{\lambda};\mathfrak{g}\right)  $, we consider the
{\bfseries\itshape exterior covariant derivative} $D^{\widetilde{W}%
}:=\mathbf{d}\circ\operatorname*{Hor}\nolimits_{\widetilde{W}}$ which
satisfies%
\[
D^{\widetilde{W}}\circ\Phi=\Phi\circ D^{\omega}%
\]
(\cite[Proposition 2]{the}). In particular, the field strength $\Phi
\circ\Omega^{\omega}$ equals $\Omega^{\widetilde{W}}:=D^{\widetilde{W}}%
\circ\widetilde{W}=\mathbf{d}\widetilde{W}+\frac{1}{2}[\widetilde
{W},\widetilde{W}]$ and%
\[
\Omega^{\widetilde{W}}(z_{1},z_{2})=[W(\xi_{1}),W(\xi_{2})]-W([\xi_{1},\xi
_{2}]),
\]
where $z_{i}=(\xi_{i},\eta_{i})\in\mathfrak{k}\times\mathfrak{g}$, $i=1,2$.

The field strength $\Omega^{\omega}$ is a $K$-invariant {\bfseries\itshape
horizontal form}, that is, it vanishes when contracted with any vector field
taking values on the vertical space. It can also be checked that the image of
horizontals forms $\Omega_{equiv}\left(  P_{\lambda};\mathfrak{g}\right)
^{\operatorname*{Hor}}$ under $\Phi$ are those chains in $\Lambda
_{\widetilde{\mathfrak{h}}}\left(  \mathfrak{k}\times\mathfrak{g}%
;\mathfrak{g}\right)  $ which only depend on elements in the horizontal space
$\operatorname*{Hor}\nolimits_{\widetilde{W}}\left(  \mathfrak{k}%
\times\mathfrak{g}\right)  $. Suppose that $K/H$ is a symmetric space such
that $\mathfrak{k}=\mathfrak{h}\oplus\mathfrak{m}$, $\left[  \mathfrak{h}%
,\mathfrak{h}\right]  \subseteq\mathfrak{h}$, $\left[  \mathfrak{h}%
,\mathfrak{m}\right]  \subseteq\mathfrak{m}$, and $\left[  \mathfrak{m}%
,\mathfrak{m}\right]  \subseteq\mathfrak{h}$. Observe that
$\operatorname*{Hor}\nolimits_{\widetilde{W}}\left(  \xi,\eta\right)  =\left(
\xi,T_{e}\lambda\left(  \xi\right)  \right)  \in\widetilde{\mathfrak{h}}$ if
$\xi\in\mathfrak{h}$. Then, since the chains in $\Lambda_{\widetilde
{\mathfrak{h}}}\left(  \mathfrak{k}\times\mathfrak{g};\mathfrak{g}\right)  $
vanish on $\widetilde{\mathfrak{h}}$, we can therefore identify $\Lambda
_{\widetilde{\mathfrak{h}}}\left(  \mathfrak{k}\times\mathfrak{g}%
;\mathfrak{g}\right)  $ with the space $\Lambda_{\mathfrak{h}}\left(
\mathfrak{m};\mathfrak{g}\right)  $ of $\mathfrak{g}$-valued chains on
$\mathfrak{m}$ such that, if $\varphi\in\Lambda_{\mathfrak{h}}^{r}\left(
\mathfrak{m};\mathfrak{g}\right)  $,%
\[
\left[  T_{e}\lambda(\xi),\varphi\left(  \upsilon_{1},...,\upsilon_{r}\right)
\right]  =\sum_{i=1}^{n}\varphi\left(  \upsilon_{1},...,[\xi,\upsilon
_{i}],...,\upsilon_{n}\right)
\]
where $\xi\in\mathfrak{h}$ and $\left\{  \upsilon_{1},...,\upsilon
_{r}\right\}  \subset\mathfrak{m}$ (see (\ref{monopole 9})).

\begin{example}
\label{ejemplo reductivo}\normalfont If $K/H$ is reductive, then
$\mathfrak{k}=\mathfrak{h}\oplus\mathfrak{m}$ and $\operatorname*{Ad}%
\nolimits_{h}(\mathfrak{m})\subseteq\mathfrak{m}$ for any $h\in H$. The field
strength $\Omega^{\mathbf{\widetilde{W}}}$ associated to the Wang map
(canonical connection)%
\[
\mathbf{W}(\xi)=\left\{
\begin{array}
[c]{l}%
T_{e}\lambda(\xi)\text{ if }\xi\in\mathfrak{h}\medskip\\
0\text{ if }\xi\in\mathfrak{m}.
\end{array}
\right.
\]
is given by $\Omega^{\mathbf{\widetilde{W}}}\left(  \upsilon_{1},\upsilon
_{2}\right)  =-T_{e}\lambda\left(  \operatorname*{proj}_{\mathfrak{h}%
}([\upsilon_{1},\upsilon_{2}])\right)  $, $\upsilon_{1},\upsilon_{2}%
\in\mathfrak{m}$. \ \ \ \ $\blacksquare$
\end{example}

\subsection{Yang-Mills equations on symmetric
spaces\label{subsection YM equations symmetric spaces}}

We are going to show that the curvature associated to the canonical connection
on a symmetric space satisfies the Yang-Mills equations (Proposition
\ref{prop canonical is YM}). Thus, let $M=K/H$ be a homogeneous symmetric
space, $\mathfrak{k}=\mathfrak{h}\oplus\mathfrak{m}$, and let $P_{\lambda}$ be
a homogeneous principal bundle given by the Lie group homomorphism
$\lambda:H\rightarrow G$. The left $K$-action $L_{P_{\lambda}}$
(\ref{monopole 2}) on $P_{\lambda}$ induces a natural $K$-action on
$\Omega_{equiv}(P_{\lambda};\mathfrak{g})^{\operatorname*{Hor}}$ by means of
the pull-backs $\left(  L_{P_{\lambda}}\right)  _{k}^{\ast}$, $k\in K$, and
hence on $\Omega\left(  M;\operatorname*{Ad}\left(  P_{\lambda}\right)
\right)  $ by the identification $\Omega\left(  M;\operatorname*{Ad}\left(
P_{\lambda}\right)  \right)  =\Omega_{equiv}(P_{\lambda};\mathfrak{g}%
)^{\operatorname*{Hor}}$. If the Riemann metric on $K/H$ and its associated
volume form are $K$-invariant, so is the product $\left\langle \cdot
,\cdot\right\rangle $ and the Hodge operator $\ast$ commutes with the
$K$-action. That is,%
\[
\ast\left(  k\cdot\varphi\right)  =k\cdot\left(  \ast\varphi\right)
\]
for any $k\in K$ and any $\varphi\in\Omega\left(  K/H;\operatorname*{Ad}%
\left(  P_{\lambda}\right)  \right)  $ (see \cite[Subsection 2.6]{the}).
Consequently, $\ast$ preserves the space of $K$-invariant forms $\Omega\left(
K/H;\operatorname*{Ad}\left(  P_{\lambda}\right)  \right)  ^{K}$. Since
$P_{\lambda}\cong\left.  (K\times G)\right/  \widetilde{H}$ and $\Phi
:\Omega_{equiv}(\left.  (K\times G)\right/  \widetilde{H};\mathfrak{g)}%
^{\operatorname*{Hor}}\cong\Lambda_{\mathfrak{h}}\left(  \mathfrak{m}%
;\mathfrak{g}\right)  $, this implies that the Hodge operator $\ast$ can be
carried to $\Lambda_{\mathfrak{h}}\left(  \mathfrak{m};\mathfrak{g}\right)  $
simply imposing that%
\begin{equation}
\Phi\circ\ast=\ast\circ\Phi, \label{monopole 11}%
\end{equation}
where we also denote the new Hodge operator in the right hand side of
(\ref{monopole 11}) by $\ast$. Additionally, the covariant codifferential
$\delta^{\omega}$ is $K$-invariant as well for any $\omega\in\Omega
_{equiv}^{1}\left(  P_{\lambda};\mathfrak{g}\right)  $ and, since both the
Hodge operator and the covariant derivative commute with $\Phi$, the operator%
\[%
\begin{array}
[c]{rrl}%
\delta_{\widetilde{W}}:\Lambda_{\mathfrak{h}}\left(  \mathfrak{m}%
;\mathfrak{g}\right)  & \longrightarrow & \Lambda_{\mathfrak{h}}\left(
\mathfrak{m};\mathfrak{g}\right) \\
\varphi & \longmapsto & -(-1)^{n(\left\vert \varphi\right\vert +1)}\ast\circ
D^{\widetilde{W}}\circ\ast\varphi,
\end{array}
\]
where $\widetilde{W}=\Phi\left(  \omega\right)  $, is such that%
\[
\Phi\circ\delta^{\omega}=\delta_{\widetilde{W}}\circ\Phi.
\]
Then, $\widetilde{W}=\Phi\left(  \omega\right)  $ is a Yang-Mills connection
if and only if%
\begin{equation}
\delta_{\widetilde{W}}\Omega^{\widetilde{W}}=0. \label{monopole 12}%
\end{equation}

In the following paragraphs, we are going to introduce some notation and carry
out a few computations that will be useful later when working out some
examples. In particular, we will justify (\ref{monopole 12}) explicitly for
the canonical connection on symmetric spaces and will explicitly exhibit the
field strength of Example \ref{ejemplo reductivo} for the homogeneous
principal bundles given in Examples \ref{examples 1}.

Recall that being $M=K/H$ symmetric, the Lie algebra $\mathfrak{k}$ can be
decomposed as $\mathfrak{k}=\mathfrak{h}\oplus\mathfrak{m}$ such that $\left[
\mathfrak{h},\mathfrak{h}\right]  \subseteq\mathfrak{h}$, $\left[
\mathfrak{h},\mathfrak{m}\right]  \subseteq\mathfrak{m}$, and $\left[
\mathfrak{m},\mathfrak{m}\right]  \subseteq\mathfrak{h}$. Let $n=\dim\left(
K/H\right)  =\dim\left(  \mathfrak{m}\right)  $. Let $\left\{  \xi_{1}%
,...,\xi_{\dim\left(  \mathfrak{h}\right)  }\right\}  $ be a basis of
$\mathfrak{h}$ and $\left\{  \upsilon_{1},...,\upsilon_{n}\right\}  $ be a
basis of $\mathfrak{m}$. The commutation relations between the elements of the
basis of $\mathfrak{h}$ and $\mathfrak{m}$ can be written as%
\[
\left[  \xi_{\alpha},\xi_{\beta}\right]  =\sum_{\gamma=1}^{\dim(\mathfrak{h}%
)}c_{~\alpha\beta}^{\gamma}\xi_{\gamma},~~~\left[  \xi_{\alpha},\upsilon
_{i}\right]  =\sum_{j=1}^{n}d_{~\alpha i}^{j}\upsilon_{j},~~~\left[
\upsilon_{i},\upsilon_{j}\right]  =\sum_{\alpha=1}^{\dim(\mathfrak{h})}%
e_{~ij}^{\alpha}\xi_{\alpha}.
\]
On the other hand, let $\left\{  \eta_{1},...,\eta_{\dim\left(  \mathfrak{g}%
\right)  }\right\}  $ be a basis of $\mathfrak{g}$, the Lie algebra of the
structural group of a homogeneous principal bundle $\pi:P_{\lambda}\rightarrow
K/H$ and suppose that%
\[
\left[  \eta_{a},\eta_{b}\right]  =\sum_{c=1}^{\dim(\mathfrak{g})}r_{~ab}%
^{c}\eta_{c}.
\]
The dual basis associated to $\left\{  \xi_{1},...,\xi_{\dim\left(
\mathfrak{h}\right)  }\right\}  $, $\left\{  \upsilon_{1},...,\upsilon
_{n}\right\}  $, and $\left\{  \eta_{1},...,\eta_{\dim\left(  \mathfrak{g}%
\right)  }\right\}  $ will be denoted with the same greek letters with upper
indices, that is, $\left\{  \xi^{1},...,\xi^{\dim\left(  \mathfrak{h}\right)
}\right\}  $, $\left\{  \upsilon^{1},...,\upsilon^{n}\right\}  $, and
$\left\{  \eta^{1},...,\eta^{\dim\left(  \mathfrak{g}\right)  }\right\}  $
respectively. The field strength $\Omega^{\mathbf{\widetilde{W}}}$ associated
to the canonical connection (see Example \ref{ejemplo reductivo}) can be
written as%
\begin{equation}
\Omega_{\lambda}^{\mathbf{\widetilde{W}}}\left(  \upsilon^{i},\upsilon
^{j}\right)  =-T_{e}\lambda\left(  \sum_{\alpha=1}^{\dim(\mathfrak{h})}%
e_{~ij}^{\alpha}\xi_{\alpha}\right)  =-\sum_{\alpha=1}^{\dim(\mathfrak{h}%
)}\sum_{a=1}^{\dim(\mathfrak{g})}e_{~ij}^{\alpha}\lambda_{\alpha}^{a}\eta_{a},
\label{monopole 17}%
\end{equation}
$\upsilon^{i},\upsilon^{j}\in\mathfrak{m}$, where $\left(  \lambda_{\alpha
}^{a}\right)  _{\alpha=1,...,\dim(\mathfrak{h})}^{a=1,...,\dim(\mathfrak{g})}$
denotes the matrix of $T_{e}\lambda$ in the basis $\{\xi_{1},...,\xi
_{\dim(\mathfrak{h})}\}$ and $\{\eta_{1},...,\eta_{\dim(\mathfrak{g})}\}$. In
(\ref{monopole 17}) we have made the dependence of $\Omega_{\lambda
}^{\mathbf{\widetilde{W}}}$ with the homomorphism $\lambda$ explicit.

$K$-invariant metrics on $K/H$ are in one-to-one correspondence with
$\operatorname*{Ad}(H)$-invariant scalar products on $\mathfrak{m}$
(\cite[Chapter X Proposition 3.1]{kobayashi 2}). Similarly, $K$-invariant
volume forms on $K/H$ correspond to $\operatorname*{Ad}(H)$-invariant volume
forms on $\mathfrak{m}$. So let $\mathbf{h}_{\mathfrak{m}}$ be the scalar
product on $\mathfrak{m}$ inducing our Riemann structure on $K/H$ and let
$\mu$ be its corresponding volume element (we are not going to differentiate
between the volume element on $K/H$ and $\mathfrak{m}$). The metric
$\mathbf{h}_{\mathfrak{m}}$ yields the musical isomorphism%
\[%
\begin{array}
[c]{rrl}%
\flat:\mathfrak{m} & \longrightarrow & \mathfrak{m}^{\ast}\\
\upsilon & \longmapsto & \mathbf{h}_{\mathfrak{m}}(\upsilon,\cdot),
\end{array}
\]
whose inverse will be denoted by $\#:\mathfrak{m}^{\ast}\rightarrow
\mathfrak{m}$. The musical isomorphisms will be used to lower and raise
indices as it is customary in physics. For example, if $\{\varphi_{i}%
^{a}\}_{i=1,...,n}^{a=1,...,\dim(\mathfrak{g})}$ are the components of the
$\mathfrak{g}$-valued one form $\varphi\in\Lambda^{1}\left(  \mathfrak{m}%
;\mathfrak{g}\right)  $, $\varphi=\sum_{i=1}^{n}\sum_{a=1}^{\dim
(\mathfrak{g})}\varphi_{~i}^{a}\upsilon^{i}\otimes\eta_{a}$, then
$\{\varphi^{ai}\}_{i=1,...,n}^{a=1,...,\dim(\mathfrak{g})}$ will be the
components of $\varphi^{\#}\in\Lambda^{1}\left(  \mathfrak{m}^{\ast
};\mathfrak{g}\right)  $, $\varphi^{\#}=\sum_{i=1}^{n}\sum_{a=1}%
^{\dim(\mathfrak{g})}\varphi^{ai}\upsilon_{i}\otimes\eta_{a}$. That is,
$\varphi^{ai}=\sum_{j=1}^{n}h^{ij}\varphi_{~j}^{a}$, where $\left(
h^{ij}\right)  _{i,j=1,...,n}$ is the inverse matrix of $\left(
h_{ij}\right)  _{i,j=1,...,n}$, $h_{ij}=\mathbf{h}_{\mathfrak{m}}\left(
\upsilon_{i},\upsilon_{j}\right)  $. It is worth noticing that, in principle,
the elements of the dual basis $\{\upsilon^{1},...,\upsilon^{n}\}$ do not
correspond to $\{\mathbf{h}_{\mathfrak{m}}(\upsilon_{1},\cdot),...,\mathbf{h}%
_{\mathfrak{m}}(\upsilon_{n},\cdot)\}$. In other words, $\upsilon^{i}$ needs
not be $\mathbf{h}_{\mathfrak{m}}(\upsilon_{i},\cdot)$, $i=1,...,n$. In order
to solve this situation and avoid a confusing notation, we may suppose that
$\{\upsilon^{1},...,\upsilon^{n}\}$ is an orthonormal basis with respect to
$\mathbf{h}_{\mathfrak{m}}$. Then, $\left(  h_{ij}\right)  _{i,j=1,...,n}$
equals the identity matrix.

Finally, let $\varphi\in\Lambda_{\mathfrak{h}}\left(  \mathfrak{m}%
;\mathfrak{g}\right)  $ be expressed in the form%
\[
\varphi=\sum_{a=1}^{\dim(\mathfrak{g})}\sum_{i_{1},...,i_{r}}^{n}%
\varphi_{~i_{1}...i_{r}}^{a}(\upsilon^{i_{1}}\wedge...\wedge\upsilon^{i_{r}%
})\otimes\eta_{a}.
\]
It is shown in \cite{the} that
\[
\left(  \ast\varphi\right)  _{~j_{1}...j_{n-r}}^{b}=\frac{1}{r!}\left\vert
\mu\right\vert ^{1/2}\sum_{i_{1},...,i_{r}=1}^{n}\varphi^{bi_{1}...i_{r}%
}\epsilon_{i_{1}...i_{r}j_{1}...j_{n-r}}%
\]
where $\left\vert \mu\right\vert =\det(\mu)\neq0$, $\epsilon$ is the
completely antisymmetric Levi-Civita symbol, and the indices $i_{1},...,i_{r}$
have been raised with $\#$. Moreover, if $W$ is the Wang map associated to the
principal connection $\omega\in\Omega_{equiv}^{1}\left(  P_{\lambda
};\mathfrak{g}\right)  $ then, for any $\varphi\in\Lambda_{\mathfrak{h}%
}\left(  \mathfrak{m};\mathfrak{g}\right)  $,%
\begin{equation}
\left(  \delta_{\widetilde{W}}\varphi\right)  \left(  \zeta_{1},...,\zeta
_{r}\right)  =-\sum_{i=1}^{n}\left[  \left(  \left.  W\right\vert
_{\mathfrak{m}}\right)  ^{\#}(\upsilon^{i}),\varphi\left(  \upsilon_{i}%
,\zeta_{1},...,\zeta_{r}\right)  \right]  , \label{monopole 13}%
\end{equation}
$\zeta_{1},...,\zeta_{r}\in\mathfrak{m}$ (\cite[Example 2.13]{the}).

\begin{proposition}
\label{prop canonical is YM}The canonical connection on a symmetric space is a
Yang-Mills connection.
\end{proposition}

\begin{proof}
The canonical connection satisfies $\left.  \mathbf{W}\right\vert
_{\mathfrak{m}}=0$. We see from (\ref{monopole 13}) that $\delta
_{\mathbf{\widetilde{W}}}=0$. Consequently $\delta_{\mathbf{\widetilde{W}}%
}\Omega^{\mathbf{\widetilde{W}}}=0$ and $\boldsymbol{\omega}=\Phi
^{-1}(\widetilde{\mathbf{W}})$ is Yang-Mills.
\end{proof}

\subsection{Examples: invariant field strengths on the
sphere\label{examples field strength}}

We want to compute in this subsection the curvature associated to the
canonical connection for the principal bundles described in Examples
\ref{examples 1}. Recall that they were principal bundles over the sphere
$\mathbb{S}^{n}$ for some values of $n\in\mathbb{N}$ and several gauge groups
$G$. These curvatures will be useful later on in order to calculate the charge
of the monopole for some of the classical examples found in the literature
(Section \ref{seccion ejemplos}).

Let $\mathfrak{k}=\mathfrak{so}(n+1)=\mathfrak{h}\oplus\mathfrak{m}$, where
$\mathfrak{h}=\mathfrak{so}(n)$ are the $\left(  n+1\right)  \times\left(
n+1\right)  $ matrices of the form%
\[%
\begin{pmatrix}
0 & 0\\
0 & B
\end{pmatrix}
,~~B\text{ skew-symmetric of degree }n,
\]
and $\mathfrak{m}$ is the subspace of all matrices of the form%
\begin{equation}%
\begin{pmatrix}
0 & -\upsilon^{\top}\\
\upsilon & 0
\end{pmatrix}
, \label{monopole 21}%
\end{equation}
where $\upsilon$ is a (column) vector in $\mathbb{R}^{n}$. Let $\left\{
\xi_{\alpha,\beta}\right\}  _{\alpha>\beta}$, $\alpha,\beta\in\{1,...,n\}$, be
the basis of $\mathfrak{so}(n)$ such that $\xi_{\alpha,\beta}$ is the matrix
whose entries are $1$ in the position $\left(  \alpha,\beta\right)  $, $-1$ in
the position $\left(  \beta,\alpha\right)  $, and $0$ elsewhere. Observe that,
for the sake of a clearer notation, we label the basis of $\mathfrak{h}$ with
two indices instead of a single one. Let $\{\upsilon_{1},...,\upsilon_{n}\}$
be the canonical basis of $\mathbb{R}^{n}$, $\upsilon_{i}=(0,\overset
{i-1)}{...},1,0,...,0)$, which is also a basis of $\mathfrak{m}$ using the
correspondence given by (\ref{monopole 21}). Then,%
\[
\lbrack\upsilon_{i},\upsilon_{j}]=\upsilon_{i}\upsilon_{j}-\upsilon
_{j}\upsilon_{i}=\xi_{j,i},~~i<j,
\]
Therefore, $[\upsilon_{i},\upsilon_{j}]=\sum_{\alpha,\beta}e_{~ij}%
^{\alpha\beta}\xi_{\alpha,\beta}$ implies $e_{~ij}^{\alpha\beta}=1$ if
$\alpha=j$ and $\beta=i$ and $0$ otherwise. In order to be coherent with our
notation, we set $\xi_{i,j}=-\xi_{j,i}$ whenever $i>j$. Thus%
\begin{equation}
\Omega_{\lambda}^{\mathbf{\widetilde{W}}}\left(  \upsilon_{i},\upsilon
_{j}\right)  =-\sum_{a=1}^{\dim(\mathfrak{g})}\lambda_{ji}^{a}\eta_{a}.
\label{monopole 22}%
\end{equation}
Let us particularize the field strength (\ref{monopole 22}) for those gauge
groups $G$ given in Examples \ref{examples 1}.

\begin{examples}
\normalfont \ \ 

\begin{enumerate}
\item[\textbf{(i)}] $n=2$ and $G=U(1)$. Here $\mathfrak{g}=\mathfrak{h}%
=\mathfrak{u}(1)$ are both isomorphic to $i\mathbb{R}$, so $\dim\left(
\mathfrak{g}\right)  =\dim(\mathfrak{h})=1$. On the other hand, $\mathfrak{m}%
=\mathbb{R}^{2}$. Let $\lambda_{m}:U(1)\rightarrow U(1)$ given by $\lambda
_{m}(\operatorname*{e}\nolimits^{iz})=\operatorname*{e}\nolimits^{izm}$. Then
$T_{e}\lambda_{m}:i\mathbb{R}\rightarrow i\mathbb{R}$ equals multiplying by
$m$ and $\Omega_{\lambda_{m}}^{\mathbf{\widetilde{W}}}\left(  \upsilon
_{1},\upsilon_{2}\right)  =-im\in i\mathbb{R}\cong\mathfrak{u}(1)$.

\item[\textbf{(ii)}] $n=4$ and $G=SO(3)$. As we have already seen, the set
$\mathcal{R}\left(  4,SO(3)\right)  $ contains three elements. The trivial
homomorphism $\lambda_{trivial}:SO(4)\rightarrow SO(3)$ sends any $h\in SO(4)$
to $e=\operatorname*{Id}\in SO(3)$, $T_{e}\lambda_{trivial}=0$, and
consequently the corresponding fields strength $\Omega_{\lambda_{trivial}%
}^{\mathbf{\widetilde{W}}}=0$ vanishes identically.

Let $\lambda_{l}:SO(4)\rightarrow SO(3)$, $l=1,2$, be the homomorphism induced
by $\sigma_{l}:S^{3}\times S^{3}\rightarrow S^{3}$ respectively such that
$\sigma_{1}\left(  x,y\right)  =x$ and $\sigma_{2}\left(  x,y\right)  =y$. One
can prove that $\mathfrak{so}(4)=\mathfrak{so}(3)^{(1)}\times\mathfrak{so}%
(3)^{(2)}$, where $\mathfrak{so}(3)^{(l)}$ is the subalgebra spanned by
$\{A^{l},B^{l},C^{l}\}$, $l=1,2$, such that $A^{l}=-\xi_{2,1}+\left(
-1\right)  ^{l}\xi_{4,3}$, $B^{l}=-\xi_{3,2}+\left(  -1\right)  ^{l}\xi_{4,1}%
$, and $C^{l}=-\xi_{3,1}+\left(  -1\right)  ^{l+1}\xi_{4,2}$ (see
\cite[Section 3]{Itoh}. The different sign in our expressions is due to a
different choice of the basis $\{\xi_{\alpha,\beta}\}_{\alpha>\beta}$ of
$\mathfrak{so}(4)$). Our initial basis can be written in terms of
$\{A^{l},B^{l},C^{l}\}$, $l=1,2$, as%
\begin{align}
\xi_{2,1}  &  =-\frac{1}{2}\left(  A^{1}+A^{2}\right)  ,~~\xi_{3,1}=-\frac
{1}{2}\left(  C^{1}+C^{2}\right)  ,~~\xi_{3,2}=-\frac{1}{2}\left(  B^{1}%
+B^{2}\right) \nonumber\\
\xi_{4,1}  &  =-\frac{1}{2}\left(  B^{1}-B^{2}\right)  ,~~\xi_{4,2}=\frac
{1}{2}\left(  C^{1}-C^{2}\right)  ,~~\xi_{4,3}=-\frac{1}{2}\left(  A^{1}%
-A^{2}\right)  . \label{monopole 36}%
\end{align}
Both $\left\{  A^{1},B^{1},C^{1}\right\}  $ and $\{A^{2},B^{2},C^{2}\}$ can be
regarded as basis of $\mathfrak{so}(3)$. The field strengths $\Omega
_{\lambda_{l}}^{\mathbf{\widetilde{W}}}:\mathbb{R}^{4}\times\mathbb{R}%
^{4}\rightarrow\mathfrak{so}\left(  3\right)  $ satisfy%
\begin{align*}
\Omega_{\lambda_{l}}^{\mathbf{\widetilde{W}}}\left(  \upsilon_{1},\upsilon
_{2}\right)   &  =\frac{1}{2}A^{l},~\Omega_{\lambda_{l}}^{\mathbf{\widetilde
{W}}}\left(  \upsilon_{1},\upsilon_{3}\right)  =\frac{1}{2}C^{l}%
,~\Omega_{\lambda_{l}}^{\mathbf{\widetilde{W}}}\left(  \upsilon_{1}%
,\upsilon_{4}\right)  =\frac{(-1)^{l+1}}{2}B^{l}\\
\Omega_{\lambda_{l}}^{\mathbf{\widetilde{W}}}\left(  \upsilon_{2},\upsilon
_{3}\right)   &  =\frac{1}{2}B^{l},~\Omega_{\lambda_{l}}^{\mathbf{\widetilde
{W}}}\left(  \upsilon_{2},\upsilon_{4}\right)  =\frac{(-1)^{l}}{2}%
C^{l},~\Omega_{\lambda_{l}}^{\mathbf{\widetilde{W}}}\left(  \upsilon
_{3},\upsilon_{4}\right)  =\frac{(-1)^{l+1}}{2}A^{l},
\end{align*}
$l=1,2$.

\item[\textbf{(iii)}] $n=4$ and $G=SO(4)$. In this example, $\mathcal{R}%
\left(  4,SO(4)\right)  $ contains $5$ elements. The trivial homomorphism has
associated a zero field strength. The identity $\lambda_{\operatorname*{Id}%
}:SO(4)\rightarrow SO(4)$ has tangent map $T_{e}\lambda_{\operatorname*{Id}%
}=\left.  \operatorname*{Id}\right\vert _{\mathfrak{so}(4)}$. Thus,
$\Omega_{\lambda_{\operatorname*{Id}}}^{\mathbf{\widetilde{W}}}\left(
\upsilon_{i},\upsilon_{j}\right)  =-\xi_{ji}$, $\upsilon_{i},\upsilon_{j}%
\in\mathbb{R}^{4}$, $i<j$.

Let $\lambda_{i}:SO(4)\rightarrow SO(4)$, $i=3,4$, be the homomorphism induced
by $\sigma_{i}:S^{3}\times S^{3}\rightarrow S^{3}\times S^{3}$ respectively
such that $\sigma_{3}\left(  x,y\right)  =(x,x)$ and $\sigma_{4}\left(
x,y\right)  =(y,y)$. Let $\{A^{1},B^{1},C^{1},A^{2},B^{2},C^{2}\}$ be the
basis of $\mathfrak{so}(4)$ introduced in \textbf{(ii)}. We are going to
consider $\mathfrak{so}(4)$ as $\mathfrak{so}(3)\times\mathfrak{so}(3)$ and
both $\{A^{1},B^{1},C^{1}\}$ and $\{A^{2},B^{2},C^{2}\}$ indistinguishably as
bases of $\mathfrak{so}(3)$. Then, using (\ref{monopole 36}),%
\begin{align*}
\Omega_{\lambda_{i}}^{\mathbf{\widetilde{W}}}\left(  \upsilon_{1},\upsilon
_{2}\right)   &  =\frac{1}{2}(A^{i},A^{i}),~~\Omega_{\lambda_{i}%
}^{\mathbf{\widetilde{W}}}\left(  \upsilon_{1},\upsilon_{3}\right)  =\frac
{1}{2}(C^{i},C^{i}),\\
\Omega_{\lambda_{i}}^{\mathbf{\widetilde{W}}}\left(  \upsilon_{1},\upsilon
_{4}\right)   &  =\frac{(-1)^{i+1}}{2}(B^{i},B^{i}),~~\Omega_{\lambda_{i}%
}^{\mathbf{\widetilde{W}}}\left(  \upsilon_{2},\upsilon_{3}\right)  =\frac
{1}{2}(B^{i},B^{i}),\\
\Omega_{\lambda_{i}}^{\mathbf{\widetilde{W}}}\left(  \upsilon_{2},\upsilon
_{4}\right)   &  =\frac{\left(  -1\right)  ^{i}}{2}(C^{i},C^{i}),~~\Omega
_{\lambda_{i}}^{\mathbf{\widetilde{W}}}\left(  \upsilon_{3},\upsilon
_{4}\right)  =\frac{(-1)^{i+1}}{2}(A^{i},A^{i}),
\end{align*}
$i=1,2$. Finally, let $\lambda_{\delta}:SO(4)\rightarrow SO(4)$ be the
homomorphism induced by $\delta:S^{3}\times S^{3}\rightarrow S^{3}\times
S^{3}$, $\delta\left(  x,y\right)  =\left(  y,x\right)  $. In this case,%
\begin{align*}
\Omega_{\lambda_{\delta}}^{\mathbf{\widetilde{W}}}\left(  \upsilon
_{1},\upsilon_{2}\right)   &  =\frac{1}{2}(A^{2},A^{1}),~~\Omega
_{\lambda_{\delta}}^{\mathbf{\widetilde{W}}}\left(  \upsilon_{1},\upsilon
_{3}\right)  =\frac{1}{2}(C^{2},C^{1}),\\
\Omega_{\lambda_{\delta}}^{\mathbf{\widetilde{W}}}\left(  \upsilon
_{1},\upsilon_{4}\right)   &  =\frac{1}{2}(-B^{2},B^{1}),~~\Omega
_{\lambda_{\delta}}^{\mathbf{\widetilde{W}}}\left(  \upsilon_{2},\upsilon
_{3}\right)  =\frac{1}{2}(B^{2},B^{1}),\\
\Omega_{\lambda_{\delta}}^{\mathbf{\widetilde{W}}}\left(  \upsilon
_{2},\upsilon_{4}\right)   &  =\frac{1}{2}(C^{2},-C^{1}),~~\Omega
_{\lambda_{\delta}}^{\mathbf{\widetilde{W}}}\left(  \upsilon_{3},\upsilon
_{4}\right)  =\frac{1}{2}(-A^{2},A^{1}).
\end{align*}

\item[\textbf{(iv)}] $n=2k\geq6$ and $G=SO(2k)$, $k\in\mathbb{N}$. As in the
item \textbf{(ii)}, $\mathcal{R}\left(  2k,SO(2k)\right)  $ contains three
elements. The trivial homomorphism and the identity $\lambda
_{\operatorname*{Id}}:SO(2k)\rightarrow SO(2k)$ are similar to \textbf{(iii)}.
The other element in $\mathcal{R}\left(  2k,SO(2k)\right)  $ is the
conjugation $\delta:SO(2k)\rightarrow SO(2k)$ by the diagonal matrix with
entries $\left(  -1,...,-1,1\right)  $. The tangent map $T_{e}\delta
:\mathfrak{so}(2k)\rightarrow\mathfrak{so}(2k)$ acts on the basis
$\{\xi_{\alpha,\beta}\}_{\alpha>\beta}$, $\alpha,\beta\in\{1,...,2k\}$, as
follows%
\[
T_{e}\delta\left(  \xi_{\alpha,\beta}\right)  =\left\{
\begin{array}
[c]{c}%
\xi_{\alpha,\beta}\text{ if }\alpha\neq2k\medskip\\
-\xi_{\alpha,\beta}\text{ if }\alpha=2k.
\end{array}
\right.
\]
Thus, $\Omega_{\delta}^{\mathbf{\widetilde{W}}}\left(  \upsilon_{i}%
,\upsilon_{j}\right)  =-\xi_{j,i}$ if $i<j$ and $j\neq2k$ and $\Omega_{\delta
}^{\mathbf{\widetilde{W}}}\left(  \upsilon_{i},\upsilon_{2k}\right)
=\xi_{2k,i}$. \ \ \ \ $\blacksquare$
\end{enumerate}
\end{examples}

\section{The Chern-Weil homomorphism. Characteristic
classes\label{section Chern classes}}

This section aims at recalling the concept of characteristic class and how the
Chern-Weil homomorphism works. Roughly speaking, given a principal bundle
$\pi:P\rightarrow M$, the Chern-Weil homomorphism associates an even
differential form on $M$ to the curvature $\Omega^{\omega}$. In order to do
that, a symmetric $\operatorname*{Ad}$-invariant polynomial on $\mathfrak{g}$
is required so that the dependence of $\Omega^{\omega}$ on the gauge indices
can be removed. The most remarkable point is that the differential form on $M$
defines a \textit{de Rham} cohomology class which is independent of the
principal connection $\omega\in\Omega^{1}\left(  P;\mathfrak{g}\right)  $
under consideration. When its degree matches the dimension of $M$, the
integral of such form over $M$ defines a topological quantity that is
interpreted as the charge of the configuration described by $\pi:P\rightarrow
M$.

Let $G$ be a Lie group with Lie algebra $\mathfrak{g}$. Let $S^{k}\left(
\mathfrak{g}^{\ast}\right)  $ be the set of maps $f:\mathfrak{g}\times
\overset{k)}{...}\times\mathfrak{g}\rightarrow\mathbb{R}$ (or $\mathbb{C}$)
which are multilinear and symmetric. That is, $f\left(  \eta_{\sigma
(1)},...,\eta_{\sigma(k)}\right)  =f\left(  \eta_{1},...,\eta_{k}\right)  $
for any permutation $\sigma\in S_{k}$ of $k$ elements. Let $S\left(
\mathfrak{g}^{\ast}\right)  ^{G}=\oplus_{k\geq0}S^{k}\left(  \mathfrak{g}%
^{\ast}\right)  ^{G}$ be the symmetric algebra of multilinear functions on
$\mathfrak{g}$ which are $\operatorname*{Ad}$-invariant. Explicitly, $f\in
S^{k}\left(  \mathfrak{g}^{\ast}\right)  ^{G}$ if $f\in S^{k}\left(
\mathfrak{g}^{\ast}\right)  $ and%
\[
f\left(  \operatorname*{Ad}\nolimits_{g}(\eta_{1}),...,\operatorname*{Ad}%
\nolimits_{g}(\eta_{k})\right)  =f\left(  \eta_{1},...,\eta_{k}\right)
\]
for any $g\in G$, any $\eta_{1},...,\eta_{k}\in\mathfrak{g}$. For later
convenience, we remark that the algebra $S\left(  \mathfrak{g}^{\ast}\right)
^{G}$ is isomorphic to the algebra $P\left(  \mathfrak{g}^{\ast}\right)
^{G}=\oplus_{k\geq0}P^{k}\left(  \mathfrak{g}^{\ast}\right)  ^{G}$ of
$\operatorname*{Ad}$-invariant homogeneous polynomials on $\mathfrak{g}$
(\cite[Section 6.2]{naber}). This isomorphism works through \textit{the
polarization formula.} Indeed, if $f\in P^{k}\left(  \mathfrak{g}^{\ast
}\right)  ^{G}$ is an homogeneous polynomial of degree $k$, we define
$\operatorname*{Sym}(f)\in S^{k}\left(  \mathfrak{g}^{\ast}\right)  ^{G}$ as
\begin{equation}
\operatorname*{Sym}(f)\left(  \eta_{1},...,\eta_{k}\right)  =\frac{1}{k!}%
\sum_{i=0}^{k-1}\left(  -1\right)  ^{i}\sum_{j_{r}\neq j_{s}}f\left(
\eta_{j_{1}}+...+\eta_{j_{k-i}}\right)  . \label{monopole 24}%
\end{equation}
For example, if $k=3$, then%
\begin{align*}
\operatorname*{Sym}(f)\left(  \eta_{1},\eta_{2},\eta_{3}\right)   &  =\frac
{1}{6}\left[  f\left(  \eta+\eta_{2}+\eta_{3}\right)  -f\left(  \eta_{1}%
+\eta_{2}\right)  -f\left(  \eta_{1}+\eta_{3}\right)  -f\left(  \eta_{2}%
+\eta_{3}\right)  \right. \\
&  \left.  +f\left(  \eta_{1}\right)  +f\left(  \eta_{2}\right)  +f\left(
\eta_{3}\right)  \right]  .
\end{align*}

Let $\pi:P\rightarrow M$ be a principal fiber bundle with structural Lie group
$G$. Let $\omega\in\Omega_{equiv}^{1}\left(  P;\mathfrak{g}\right)  $ be a
principal connection and let $\Omega^{\omega}\in\Omega_{equiv}^{2}\left(
P;\mathfrak{g}\right)  ^{\operatorname*{Hor}}$ its curvature. If $f\in
S^{k}\left(  \mathfrak{g}^{\ast}\right)  ^{G}$, then the $2k$-form%
\[
\bar{f}\left(  \Omega^{\omega}\right)  (p)\left(  X_{1},...,X_{2k}\right)
=\frac{1}{2^{k}}\sum_{\sigma\in S_{2k}}\left(  -1\right)  ^{\left\vert
\sigma\right\vert }f\left(  \Omega^{\omega}(p)(X_{\sigma(1)},X_{\sigma
(2)}),...,\Omega^{\omega}(p)(X_{\sigma(2k-1)},X_{\sigma(2k)})\right)
\]
is $G$-invariant and horizontal. Therefore, there exists a uniquely defined
$2k$-form $cw\left(  f,P,\omega\right)  \in\Omega^{2k}\left(  M\right)  $ such
that%
\[
\pi^{\ast}(cw\left(  f,P,\omega\right)  )=\bar{f}\left(  \Omega^{\omega
}\right)  .
\]
The form $cw\left(  f,P,\omega\right)  $ is called the {\bfseries\itshape
Chern-Weil form} of $f$. What is more important, $cw\left(  f,P,\omega\right)
$ is closed, so there is a well defined \textit{de Rham} cohomology class
$[cw\left(  f,P,\omega\right)  ]\in H^{2k}\left(  M\right)  $ called the
{\bfseries\itshape characteristic class} of the invariant polynomial $f$
(\cite[Theorem 10.4.3]{bleecker}, \cite[Theorem 20.3]{michor}), which is
independent of the particular choice of $\omega\in\Omega_{equiv}%
^{1}(P;\mathfrak{g})$ (\cite[Theorem 10.4.11]{bleecker}). That is, it only
depends on the fiber bundle structure of $P$. It is worth noticing that the
proof of this fact uses that the polynomial $f$ is \textit{symmetric}. For
example, the characteristic classes of a trivial principal bundle all vanish.
In addition, the mapping%
\[%
\begin{array}
[c]{rrl}%
Cw_{P}:S\left(  \mathfrak{g}^{\ast}\right)  ^{G} & \longrightarrow & H^{\ast
}\left(  M\right) \\
f & \longmapsto & [cw\left(  f,P,\omega\right)  ]
\end{array}
\]
is a homomorphism of commutative algebras, known as the {\bfseries\itshape
Chern-Weil homomorphism}. If two principal bundles $P$ and $P^{\prime}$ over
$M$ are isomorphic, they give rise to the same Chern-Weil homomorphism
(\cite[Theorem 10.4.8]{bleecker}).

We are going to assume from now on that our Lie group $G$ is contained in
$GL\left(  m,\mathbb{R}\right)  $ for some $m\in\mathbb{N}$. Roughly speaking,
$G$ may be thought as a classical matrix Lie group. For such groups, the
adjoint action of $G$ on $\mathfrak{g}$ has a simple expression. That is,
\[
\operatorname*{Ad}\nolimits_{g}(\xi)=g\xi g^{-1},~g\in G,~\xi\in\mathfrak{g},
\]
where $g\xi g^{-1}$ is a product of matrices. For a matrix $A\in
\mathfrak{gl}(m,\mathbb{R})$, the {\bfseries\itshape characteristic
coefficient} $c_{k}^{m}(A)$ are implicitly given by the equation%
\[
\det\left(  t\operatorname*{Id}+\frac{i}{2\pi}A\right)  =\sum_{k=0}^{m}%
t^{m-k}c_{k}^{m}(A),~~t\in\mathbb{R}.
\]
The characteristic coefficients are homogeneous polynomials of degree $k$
which are $\operatorname*{Ad}$-invariant. Furthermore, they satisfy the
recursive formula (\cite[Lemma 20.9]{michor})%
\begin{equation}
c_{k}^{m}(A)=\frac{1}{k}\sum_{j=0}^{k-1}\left(  -1\right)  ^{k-j-1}\left(
\frac{i}{2\pi}\right)  ^{k-j}c_{j}^{m}(A)\operatorname*{trace}(A^{k-j}%
),~~A\in\mathfrak{gl}(m,\mathbb{R}). \label{monopole 15}%
\end{equation}
For example, it is easy to show from (\ref{monopole 15}) that%
\begin{gather*}
c_{0}^{m}\left(  A\right)  =1,~~c_{1}\left(  A\right)  =\frac{i}{2\pi
}\operatorname*{trace}A,~~c_{2}^{m}\left(  A\right)  =-\frac{1}{8\pi^{2}%
}\left[  \left(  \operatorname*{trace}A\right)  ^{2}-\operatorname*{trace}%
\left(  A^{2}\right)  \right] \\
c_{3}^{m}\left(  A\right)  =-\frac{i}{48\pi^{3}}\left[  \left(
\operatorname*{trace}A\right)  ^{3}-3\operatorname*{trace}\left(
A^{2}\right)  \operatorname*{trace}A+2\operatorname*{trace}\left(
A^{3}\right)  \right]  .
\end{gather*}
The $k$-th {\bfseries\itshape Chern class} is defined as%
\[
c_{k}\left(  P\right)  :=Cw_{P}(\operatorname*{Sym}(c_{k}^{m}))\in H^{2k}(M).
\]
Among other characteristic classes, we choose the Chern classes because,
despite the presence of the imaginary unit $i\in\mathbb{C}$ in their
definition, they are actually real cohomology classes provided that $G$ is a
subgroup of the unitary group $U\left(  m\right)  $ as in our examples
(\cite[20.13]{michor}). If we write $\Omega^{\omega}$ as a matrix valued two
form $\left(  (\Omega^{\omega})_{j}^{i}\right)  _{i,j=1,...,m}$, then%
\begin{equation}
\pi^{\ast}\left(  cw\left(  \operatorname*{Sym}(c_{k}^{m}),P,\omega\right)
\right)  =\frac{\left(  -1\right)  ^{k}}{\left(  2\pi i\right)  ^{k}k!}%
\sum_{i_{1}<...<i_{k}}\sum_{\sigma\in S_{k}}\left(  -1\right)  ^{\left\vert
\sigma\right\vert }\left(  \Omega^{\omega}\right)  _{\sigma(i_{1})}^{i_{1}%
}\wedge...\wedge\left(  \Omega^{\omega}\right)  _{\sigma(i_{k})}^{i_{k}}
\label{monopole 30}%
\end{equation}
(\cite[page 309]{kobayashi 2}).

\begin{example}
\normalfont The characteristic coefficients are usually algebraically
independent and generate the algebra of polynomial functions on $\mathfrak{g}$
invariant by $\operatorname*{Ad}\nolimits_{G}$, at least for some of the
classical matrix groups such as $U\left(  m\right)  $ (\cite[Chapter
XII]{kobayashi 2}). However, we are going to deal in this example with
$G=SO(m)$ whose Lie algebra $\mathfrak{g}$ is the algebra of skew-symmetric
matrices or order $m\in\mathbb{N}$. The characteristic coefficients $c_{k}%
^{m}$ are then equal to zero if $k$ is odd, as it can be inductively checked
from (\ref{monopole 15}). Moreover, if $m=2q+1$ is odd, then $\{c_{2}%
^{m},...,c_{2q}^{m}\}$ are indeed algebraically independent and generate
$P\left(  \mathfrak{so}(m)^{\ast}\right)  ^{SO(m)}$ (\cite[Chapter XII Theorem
2.7]{kobayashi 2}). If $m=2q$ is even, however, there exists a polynomial
function $\operatorname*{Pf}$ (unique up to a sign) such that $c_{2q}%
^{m}=\left(  -1\right)  ^{q}\left(  2\pi\right)  ^{-2q}\operatorname*{Pf}^{2}$
and the functions $\{c_{2}^{m},...,c_{2(q-1)}^{m},\operatorname*{Pf}\}$ are
algebraically independent and generate $P(\mathfrak{so}(m)^{\ast})^{SO(m)}$.
The polynomial $\operatorname*{Pf}$ is called {\bfseries\itshape the Pfaffian}
and, up to a factor, equals the square root of the determinant of a matrix. If
the matrix $A\in\mathfrak{so}(m)$ is written as $A=(A_{j}^{i})_{i,j=1,...,2q}%
$, then%
\[
\operatorname*{Pf}\left(  A\right)  =\frac{1}{2^{q}q!}\sum_{\eta\in S_{2q}%
}\left(  -1\right)  ^{\left\vert \eta\right\vert }A_{\eta(2)}^{\eta(1)}\cdots
A_{\eta(2q)}^{\eta(2q-1)}.
\]
The {\bfseries\itshape Euler class} $\chi\left(  P\right)  $ is defined as
$\frac{1}{\pi^{q}}Cw_{P}\left(  \operatorname*{Sym}(\operatorname*{Pf}%
)\right)  $. If the curvature $\Omega^{\omega}\in\Omega_{equiv}^{2}\left(
P;\mathfrak{g}\right)  ^{\operatorname*{Hor}}$ of some principal connection
$\omega\in\Omega_{equiv}^{1}\left(  P;\mathfrak{g}\right)  $ on $\pi
:P\rightarrow M$ is written as a matrix valued two form $\left(
(\Omega^{\omega})_{j}^{i}\right)  _{i,j=1,...,2q}$, then%
\begin{equation}
\frac{1}{\pi^{q}}Cw_{P}\left(  \operatorname*{Sym}(\operatorname*{Pf})\right)
=\frac{1}{2^{q}\pi^{q}q!}\sum_{\eta\in S_{2q}}\left(  -1\right)  ^{\left\vert
\eta\right\vert }\left(  \Omega^{\omega}\right)  _{\eta(2)}^{\eta(1)}%
\wedge\cdots\wedge\left(  \Omega^{\omega}\right)  _{\eta(2q)}^{\eta
(2q-1)}\label{monopole 28}%
\end{equation}
(\cite[Chapter XII Theorem 5.1]{kobayashi 2}). \ \ \ \ $\blacksquare$
\end{example}

Finally, we are going to introduce the charge of a monopole. So let
$\pi:P_{\lambda}\rightarrow\mathbb{S}^{n}$ be a homogeneous principal bundle
over the $n$-dimensional sphere. The sphere equals the symmetric space
$SO(n+1)/SO(n)$. It is a Riemann manifold with the Riemannian structure
inherited from $\mathbb{R}^{n+1}$. We know by Theorem \ref{teorema laquer}
that the canonical connection is the unique which is invariant by $SO(n+1)$.
Let $\boldsymbol{\omega}\in\Omega_{equiv}^{1}\left(  P;\mathfrak{g}\right)  $
denote the canonical connection and $\Omega^{\boldsymbol{\omega}}\in
\Omega_{equiv}^{2}\left(  P;\mathfrak{g}\right)  ^{\operatorname*{Hor}}$ its
curvature, which is $SO(n+1)$-invariant by the left translations $L_{\lambda}$
(Eq. (\ref{monopole 2})). Then $cw\left(  \operatorname*{Sym}(f),P_{\lambda
},\boldsymbol{\omega}\right)  \in\Omega^{2k}\left(  \mathbb{S}^{n}\right)  $,
$f\in P^{k}(\mathfrak{g}^{\ast})^{G}$, is also invariant by the natural left
$SO(n+1)$-action we have on $\mathbb{S}^{n}$. Suppose that $n=2q$ is even. In
that case, $cw(\operatorname*{Sym}(f),P_{\lambda},\boldsymbol{\omega})$, $f\in
P^{q}(\mathfrak{g}^{\ast})^{G}$, is proportional to the volume element $\mu$
of $\mathbb{S}^{n}$ induced from the standard metric, i.e.,%
\[
cw(\operatorname*{Sym}(f),P_{\lambda},\boldsymbol{\omega})=\mathbf{d}\mu
\]
for some function $\mathbf{d}\in C^{\infty}(\mathbb{S}^{n})$. Since $\mu$ is
also $SO(n+1)$ invariant, so is $\mathbf{d}\in C^{\infty}(\mathbb{S}^{n})$.
But the only functions on $\mathbb{S}^{n}$ which are invariant by the special
orthogonal group are the constants, so $\mathbf{d}\in\mathbb{R}$. Regarding
$\mathbb{S}^{n}$ as an imbedded submanifold of $\mathbb{R}^{n+1}$, we may
consider $\mathbf{d}$ as a function of the radius. It is worth observing that,
once $f\in P^{q}(\mathfrak{g}^{\ast})^{G}$ is given, $\mathbf{d}$ can be
easily computed from the expression of $\Omega^{\boldsymbol{\omega}}$ given in
(\ref{monopole 22}) (see examples in Section \ref{seccion ejemplos}). The
quantity%
\begin{equation}
Q:=\int_{\mathbb{S}^{n}}cw\left(  \operatorname*{Sym}(f),P_{\lambda
},\boldsymbol{\omega}\right)  =\mathbf{d}\operatorname{vol}(\mathbb{S}^{n})
\label{monopole 16}%
\end{equation}
will be called the {\bfseries\itshape charge of the monopole}. Up to a factor,
it can be interpreted as the flow of the field strength $\Omega
^{\boldsymbol{\omega}}$ trough the surface of the sphere $\mathbb{S}^{n}$.
However, in order to match the order of $\Omega^{\boldsymbol{\omega}}$ with
the dimension of $\mathbb{S}^{n}$ we need some characteristic class $cw\left(
\operatorname*{Sym}(f),P_{\lambda},\boldsymbol{\omega}\right)  $, for example
the Chern class (of suitable order). It is worth noticing that the charge of
the monopole does not depend on the fact that we have worked with the
canonical connection $\boldsymbol{\omega}$ because, as we already said, the
Chern-Weil homomorphism does not depend on $\boldsymbol{\omega}$. In other
words, it is a topological invariant. Obviously, the charge depends strongly
on the choice of the invariant polynomial $f\in P^{q}(\mathfrak{g}^{\ast}%
)^{G}$ or, equivalently, on the characteristic class $cw\left(
\operatorname*{Sym}(f),P_{\lambda},\boldsymbol{\omega}\right)  $ and, for some
$f\in P^{q}(\mathfrak{g}^{\ast})^{G}$, \textit{it could be zero even for
non-trivial bundles}. As we will discuss in the examples, we will define the
charge integrating on $\mathbb{S}^{n}$ either the Chern class $c_{q}\left(
P\right)  $, $n=2q$, or the Euler class $\chi\left(  P\right)  $ in order to
label all the non-isomorphic principal bundles over $\mathbb{S}^{n}$ with a
different value of their charge. These two classes are, up to a constant
factor, essentially the unique characteristic classes we can use to define the
charge in most classical matrix Lie groups.

\section{Examples\label{seccion ejemplos}}

\subsection{The Dirac monopole}

The first one in introducing the concept of monopole was Dirac in the context
of electromagnetic field theory \cite{Dirac}. Dirac showed that there exist
static singular solutions of the Maxwell equations on $\mathbb{R}%
^{3}\backslash\{0\}$ with a pointwise magnetic source placed at the origin
$0\in\mathbb{R}^{3}$. In order to be gauge invariant, the magnetic charge
needed to be an integer in appropriate units. Since there is no evidence of
the existence of such magnetic charge (despite the efforts carried out to find
it since then), Dirac monopoles might have seemed useless at first sight.
Nevertheless, and more importantly, the fact that the magnetic charge can only
take discrete values implies in turn that the electric charge needs do so, as
we experimentally observe. In other words, both the magnetic and the electric
charge are \textit{quantized}. Thus the relevance of such magnetic monopoles.

A free electromagnetic field is a Yang-Mills theory with gauge group $U\left(
1\right)  $. In particular, Dirac's monopoles are described as principal
bundles over $\mathbb{S}^{2}$ (that is, principal bundles over $\mathbb{R}%
^{3}\backslash\{0\}$) with structural group $U\left(  1\right)  $. Since we
require the potential vector field, and its corresponding field strength, to
be $SO\left(  3\right)  $-invariant, such principal bundles $\pi
:P_{\lambda_{m}}\rightarrow\mathbb{S}^{2}$ are in one-to-one correspondence
with the homomorphisms $\lambda_{m}:U\left(  1\right)  \rightarrow U\left(
1\right)  $, $m\in\mathbb{Z}$, introduced in Examples \ref{examples 1}
\textbf{(i)}. The $SO\left(  3\right)  $-invariant field strength
$\Omega^{\boldsymbol{\omega}}$ is built from the canonical connection and
computed at $o\in SO\left(  3\right)  /U(1)\cong\mathbb{S}^{2}$ in Subsection
\ref{examples field strength} \textbf{(i)}. Recall that $\Omega_{\lambda_{m}%
}^{\widetilde{\mathbf{W}}}\left(  \upsilon_{1},\upsilon_{2}\right)  =-im$,
where $\{\upsilon_{1},\upsilon_{2}\}$ is the canonical basis of $\mathbb{R}%
^{2}=T_{o}\mathbb{S}^{2}$. The charge associated to these configurations is
given by integrating the first Chern class $[\frac{i}{2\pi}%
\operatorname*{trace}\Omega^{\boldsymbol{\omega}}]\in H^{2}\left(
\mathbb{S}^{2}\right)  $ over $\mathbb{S}^{2}$. As we already pointed out,
$\frac{i}{2\pi}\operatorname*{trace}\Omega^{\boldsymbol{\omega}}=\mathbf{d}%
\mu$ for some constant $\mathbf{d}\in\mathbb{R}$ and where $\mu$ is the volume
$2$-form of $\mathbb{S}^{2}$. The constant $\mathbf{d}$ can be calculated as
follows,%
\[
\mathbf{d}=\frac{i}{2\pi}\operatorname*{trace}\left(  \Omega
^{\boldsymbol{\omega}}(o)\left(  \upsilon_{1},\upsilon_{2}\right)  \right)
=\frac{i}{2\pi}\operatorname*{trace}\left(  \Omega_{\lambda_{m}}%
^{\widetilde{\mathbf{W}}}\left(  \upsilon_{1},\upsilon_{2}\right)  \right)
=\frac{m}{2\pi},
\]
and the charge%
\[
Q_{m}:=\frac{i}{2\pi}\int_{\mathbb{S}^{2}}\operatorname*{trace}\Omega
^{\boldsymbol{\omega}}=\frac{m}{2\pi}\int_{\mathbb{S}^{2}}\mu=2m.
\]

\subsection{The Yang monopole}

Yang monopoles are non-trivial solutions of Yang-Mills theories on
$\mathbb{R}^{4}\backslash\{0\}$ (equivalently on $\mathbb{S}^{4}$) with gauge
group $G=SU(2)$. Unlike the general approach throughout this paper, where we
considered the sphere as a quotient of orthogonal groups, we are now going to
regard $\mathbb{S}^{4}$ as a quotient of spin groups, i.e., $\mathbb{S}%
^{4}=Spin(5)/Spin(4)$. That is, we are going to describe principal bundles
$\pi:P\rightarrow\mathbb{S}^{4}$ with gauge group $SU(2)$ and a left $Spin(5)$
action projecting onto the $Spin(5)$ action on $Spin(5)/Spin(4)$, which
obviously coincides with the standard $SO(5)$ action on $\mathbb{S}^{4}$. In
his paper \cite{yang}, Yang describes monopole configurations on
$\mathbb{S}^{4}$ which are invariant by the standard action of $SO(5)$. In our
opinion, his description is imprecise and he should have talked about
$Spin(5)$ invariant monopoles. Indeed, as the next proposition shows, there
does not exist any non-trivial principal bundle over $\mathbb{S}^{4}$ with
gauge group $SU(2)$ supporting a $SO(5)$ left action. However, since Yang
worked with potentials and field strengths on $\mathbb{S}^{4}$ using local
sections, he did not realize that his $SO(5)$ action actually came from a
$Spin(5)$ action on the whole bundle.

\begin{proposition}
\label{prop 2}The unique homomorphism of Lie groups $\lambda:SO(4)\rightarrow
SU(2)$ from $SO(4)$ to $SU(2)$ is the trivial homomorphism, $\lambda(h)=e\in
SU(2)$ for any $h\in SO(4)$.
\end{proposition}

\begin{proof}
As we saw in Examples \ref{examples 1} \textbf{(ii)}, $Spin(4)=S^{3}\times
S^{3}$ and $SO(4)=\left(  S^{3}\times S^{3}\right)  /\{\left(  1,1\right)
,\left(  -1,-1\right)  \}$, where $S^{3}$ is the quaternionic sphere. Let
$\tau:Spin(4)\rightarrow SO(4)$ be the covering homomorphism. On the other
hand, we already argued that $SU(2)=S^{3}$. Recall that, modulo conjugation,
the unique homomorphisms between $Spin(4)$ and $SU(2)$ are the trivial one and
the projections $\sigma_{l}:S^{3}\times S^{3}\rightarrow S^{3}$, $l=1,2$, such
that $\sigma_{1}\left(  x,y\right)  =x$ and $\sigma_{2}\left(  x,y\right)
=y$, $\left(  x,y\right)  \in S^{3}\times S^{3}$.

Suppose that there exists a homomorphism $\lambda:SO(4)\rightarrow SU(2)$
different from the trivial one. Then, $\lambda\circ\tau:Spin(4)\rightarrow
SU(2)$ is conjugated to $\sigma_{1}$ or $\sigma_{2}$. Assume that it is
conjugated to $\sigma_{1}$. Therefore, there exists some $g\in SU(2)$ such
that%
\[
\lambda_{1}=g\left(  \lambda\circ\tau\right)  g^{-1}=g\lambda g^{-1}\circ
\tau.
\]
Replacing $\lambda$ with $g\lambda g^{-1}$ if necessary, we may suppose that
$\sigma_{1}=\lambda\circ\tau$, where $\lambda$ is different from the trivial
homomorphism. But this is clearly a contradiction, since $\tau(\left(
x,y\right)  )=\tau(\left(  -x,-y\right)  )\in SO(4)$ and $\sigma_{1}(\left(
x,y\right)  )\neq\sigma_{1}(\left(  -x,-y\right)  )$.\smallskip
\end{proof}

In conclusion, we have two non-trivial homogeneous principal bundles
$\pi_{\sigma_{l}}:P_{\sigma_{l}}\rightarrow\mathbb{S}^{4}$ associated to the
homomorphisms $\sigma_{l}:Spin(4)=S^{3}\times S^{3}\rightarrow SU(2)=S^{3}$,
$l=1,2$. It is worth noting that $\pi_{\sigma_{1}}:P_{\sigma_{1}}%
\rightarrow\mathbb{S}^{4}$ was already identified in \cite[Subsection
4.3]{asorey} as the principal bundle behind the BPST instanton.

We want to compute the charge $Q$ (Eq. (\ref{monopole 16})) of $\pi
_{\sigma_{l}}:P_{\sigma_{l}}\rightarrow\mathbb{S}^{4}$, $l=1,2$, by means of
the second Chern class. According to \cite{naber}, two principal bundles over
$\mathbb{S}^{4}$ and gauge group $SU(2)$ are isomorphic if and only if they
have the same Chern number. Therefore, the charge provided by the second Chern
class seems a good topological invariant to differentiate the two non-trivial
monopole configurations.

The field strengths associated to the $Spin(5)$-invariant canonical
connections of $\pi_{\sigma_{l}}:P_{\sigma_{l}}\rightarrow\mathbb{S}^{4}$,
$l=1,2$, are given in Subsection \ref{examples field strength} \textbf{(ii)}.
Indeed, $\Omega_{\lambda_{l}}^{\widetilde{\mathbf{W}}}$, $l=1,2$, in
Subsection \ref{examples field strength} \textbf{(ii) }are the curvatures
associated to the homomorphisms $\lambda_{l}:SO(4)\rightarrow SO(3)$ which, in
turn, are induced from $\sigma_{l}:Spin(4)\rightarrow SU(2)$. Since in order
to compute de curvatures of the canonical connections we only need the tangent
maps $T_{e}\lambda_{l}:\mathfrak{so}(4)\rightarrow\mathfrak{so}(3)$ and the
Lie algebras $\mathfrak{spin}(4)$ and $\mathfrak{su}(2)$ coincide with
$\mathfrak{so}(4)$ and $\mathfrak{so}(3)$ respectively, $\Omega_{\lambda_{l}%
}^{\widetilde{\mathbf{W}}}:\mathbb{R}^{4}\times\mathbb{R}^{4}\rightarrow
\mathfrak{su}(2)$ are the $Spin(5)$-invariant curvatures evaluated at
$p=[e,e]^{\sim}\in P_{\sigma_{l}}$. However, observe that these field
strengths take values in two subalgebras of $\mathfrak{so}(4)$, those
generated by the matrices $\{A^{l},B^{l},C^{l}\}$, $l=1,2$, which are
isomorphic to $\mathfrak{su}(2)$. We need to implement these isomorphisms
explicitly since, in order to compute the second Chern class using
(\ref{monopole 30}), $\mathfrak{su}(2)$ must be regarded as Lie algebra of
complex matrices contained in $\mathfrak{u}(m)$ for some $m\in\mathbb{N}$. The
easiest solution is to establish the correspondence%
\[
A^{l}\mapsto\frac{i}{2}\sigma^{1}=%
\begin{pmatrix}
0 & i/2\\
i/2 & 0
\end{pmatrix}
\text{, \ }B^{l}\mapsto\frac{i}{2}\sigma^{2}=%
\begin{pmatrix}
0 & 1/2\\
-1/2 & 0
\end{pmatrix}
\text{, \ }C^{l}\mapsto-\frac{i}{2}\sigma^{3}=%
\begin{pmatrix}
-i/2 & 0\\
0 & i/2
\end{pmatrix}
,
\]
$l=1,2$, where $\{\sigma^{1},\sigma^{2},\sigma^{3}\}$ are the Pauli matrices.

The second Chern classes $cw\left(  S(c_{2}^{2}),P_{\sigma_{l}},\omega\right)
$, $l=1,2$, are proportional to the canonical volume element $\mu$ of
$\mathbb{S}^{4}$, $cw\left(  S(c_{2}^{2}),P_{\sigma_{l}},\omega\right)
=\mathbf{d}\mu$. The constant of proportionality $\mathbf{d}$ can be obtained
as%
\[
\mathbf{d}=cw\left(  S(c_{2}^{2}),P_{\sigma_{l}},\omega\right)  \left(
\upsilon_{1},\upsilon_{2},\upsilon_{3},\upsilon_{4}\right)  ,
\]
where $\{\upsilon_{1},\upsilon_{2},\upsilon_{3},\upsilon_{4}\}$ is the
canonical orthonormal basis of $\mathbb{R}^{4}\cong T_{o}\mathbb{S}^{4}$.
Identifying $\mathbb{R}^{4}=\mathfrak{m}\subset\mathfrak{so}(4)$ with the
horizontal space $\operatorname*{Hor}_{p}$ at $p=[e,e]^{\sim}\in P_{\sigma
_{l}}$, $\mathbf{d}$ can be computed inserting the explicit expressions for
$\Omega_{\lambda_{l}}^{\widetilde{\mathbf{W}}}$ in (\ref{monopole 30}). We
prefer, however, giving directly the charge $Q=\mathbf{d}\operatorname{vol}%
\left(  \mathbb{S}^{4}\right)  $, which equals $-\frac{1}{8}$ for the
homogeneous bundle $\pi_{\sigma_{1}}:P_{\sigma_{1}}\rightarrow\mathbb{S}^{4}$
and $\frac{1}{8}$ for $\pi_{\sigma_{2}}:P_{\sigma_{2}}\rightarrow
\mathbb{S}^{4}$. These results are in complete agreement with \cite{yang} and
we therefore omit explicit computations. If Yang gave the Chern number $-1$
and $1$ respectively to these bundles was because, in his definition of the
second Chern class, he chose a coefficient $8$ times greater than ours. Since
two principal $SU(2)$-bundles over $\mathbb{S}^{4}$ are isomorphic if and only
if they have the same Chern number (\cite{naber}), the two principal bundles
with non-vanishing charge we obtained are isomorphic to Yang's.

\subsection{$SO(2n)$-monopoles.\label{subsection Townsend monopole}}

As far as we know, $SO(2n)$-monopoles seem to be appeared for the first time
in \cite{nepomechie}, where the author tries to generalize the Dirac monopole
to Kalb-Ramond fields (\cite{kalb-ramond}), although they may have been
introduced in previous works under a different appearance. Since then,
Tchrakian has studied such monopoles in depth. We recommend two of his latest
works \cite{tchrakian,tchrakian 2000}, and references therein, for an approach
to $SO(2n)$-monopoles complementary to ours.

In his paper \cite{townsend}, Gibbons and Townsend study monopole
configurations over the sphere $\mathbb{S}^{2q}$, with gauge group $SO(2q)$,
$q\geq2$. However, they only deal with the principal bundle
$SO(2q+1)\rightarrow SO(2q+1)/SO(2q)$, which corresponds to the homogeneous
principal bundle $P_{\lambda_{\operatorname*{Id}}}$ given by the identity
homomorphism $\lambda_{\operatorname*{Id}}:SO(2q)\rightarrow SO(2q)$
(\cite[Section 4]{townsend}), and exhibit the corresponding $SO(2q+1)$%
-invariant (canonical) connection. In addition, they define the charge of the
monopole as the integral over $\mathbb{S}^{2q}$ of the $2q$-form%
\begin{equation}
\operatorname*{trace}\left(  \Omega^{\boldsymbol{\omega}}\wedge\overset
{q)}{...}\wedge\Omega^{\boldsymbol{\omega}}\right)  ,\label{monopole 26}%
\end{equation}
where $\Omega^{\omega}$ is the field strength associated to the canonical
principal connection. Up to a constant factor, this characteristic class
coincides with the Chern class for the case $k=2$. In our opinion, some
authors choose (\ref{monopole 26}) to define the charge (see for instance
\cite{meng}) because it is a straightforward generalization of the integrand
$\operatorname*{trace}\left(  \Omega^{\omega}\wedge\Omega^{\omega}\right)  $
used by Yang to compute the charge of his monopole. In \cite{yang}, Yang
points out that he deliberately chooses the second Chern class. Nevertheless,
it is not clear to which $\operatorname*{Ad}\nolimits_{SO(2q)}$-invariant
polynomial $f\in P(\mathfrak{so(}2q)^{\ast})^{SO(2q)}$ corresponds the
$2q$-form (\ref{monopole 26}). Moreover, it is claimed in \cite{townsend}, but
no proof is provided, that the field strength $\Omega^{\omega}$ can be written
in a suitable basis of $\mathfrak{so}(2q)$ such that%
\begin{equation}
\int_{\mathbb{S}^{2q}}\operatorname*{trace}\left(  \Omega^{\boldsymbol{\omega
}}\wedge\overset{q)}{...}\wedge\Omega^{\boldsymbol{\omega}}\right)
\neq0.\label{monopole 27}%
\end{equation}
In our opinion, this result is not correct. The argument against
(\ref{monopole 27}) works as follows: since $\operatorname*{trace}%
(\Omega^{\boldsymbol{\omega}}\wedge\overset{q)}{...}\wedge\Omega
^{\boldsymbol{\omega}})$ is proportional to the natural volume element $\mu$
of $\mathbb{S}^{2q}$, we only need to compute the constant of proportionality
$\mathbf{d}$ in order to value (\ref{monopole 27}). Furthermore, this
computation can be carried out at any point $m\in\mathbb{S}^{2q}$ of the
sphere. If $\{\upsilon_{1},...,\upsilon_{2q}\}$ is an orthonormal basis of
$T_{m}\mathbb{S}^{2q}\cong\mathbb{R}^{2q}$, then
\[
\mathbf{d}=\operatorname*{trace}\left(  \Omega^{\boldsymbol{\omega}}%
\wedge\overset{q)}{...}\wedge\Omega^{\boldsymbol{\omega}}\right)  \left(
\upsilon_{1},...,\upsilon_{2q}\right)  .
\]
Let $o\in\mathbb{S}^{2q}$. Since the charge is a topological invariant, we can
compute it using any field strength on $SO(2q+1)\rightarrow SO(2q+1)/SO(2q)$.
According to Subsection \ref{examples field strength} \textbf{(iv)}, the field
strength at $o$ is given by $\Omega_{\lambda_{\operatorname*{Id}}%
}^{\mathbf{\widetilde{W}}}\left(  \upsilon_{i},\upsilon_{j}\right)
=-\xi_{j,i}=\xi_{i,j}\in\mathfrak{so}(2q)$. The matrix $\xi_{j,i}$ has entries%
\begin{equation}
\left(  \xi_{j,i}\right)  _{\beta}^{\alpha}=\left(  -1\right)  ^{U(i-j)}%
\left(  -1\right)  ^{U(\beta-\alpha)}\delta_{j}^{\alpha}\delta_{i\beta
},\label{monopole 31}%
\end{equation}
where $U$ is the Heaviside step function, $U(x)=1$ if $x>0$ and $U(x)=0$ if
$x\leq0$. Therefore
\begin{align*}
&  \operatorname*{trace}\left(  \Omega^{\boldsymbol{\omega}}\wedge\overset
{q)}{...}\wedge\Omega^{\boldsymbol{\omega}}\right)  \left(  \upsilon
_{1},...,\upsilon_{2q}\right)  \\
&  =\frac{1}{2^{q}}\operatorname*{trace}\left(  \sum_{\sigma\in S_{2q}}\left(
-1\right)  ^{\left\vert \sigma\right\vert }\Omega_{\lambda_{\operatorname*{Id}%
}}^{\mathbf{\widetilde{W}}}\left(  \upsilon_{\sigma(1)},\upsilon_{\sigma
(2)}\right)  \cdots\Omega_{\lambda_{\operatorname*{Id}}}^{\mathbf{\widetilde
{W}}}\left(  \upsilon_{\sigma(2q-1)},\upsilon_{\sigma(2q)}\right)  \right)  \\
&  =\frac{\left(  -1\right)  ^{q}}{2^{q}}\operatorname*{trace}\left(
\sum_{\sigma\in S_{2q}}\left(  -1\right)  ^{\left\vert \sigma\right\vert }%
\xi_{\sigma(2),\sigma(1)}\cdots\xi_{\sigma(2q),\sigma(2q-1)}\right)
\end{align*}
but $\xi_{\sigma(2),\sigma(1)}\cdots\xi_{\sigma(2q),\sigma(2q-1)}=0$ for any
$\sigma\in S_{2q}$ because the matrix product $\xi_{j,i}\xi_{r,s}$ is zero if
the indices $\left(  j,i\right)  $ are different from $\left(  r,s\right)  $.
Thus, $\operatorname*{trace}(\Omega^{\boldsymbol{\omega}}\wedge\overset
{q)}{...}\wedge\Omega^{\boldsymbol{\omega}})=0$.

The Chern class is not useful to define the monopole charge either, since it
also vanishes. The details are given in Subsection
\ref{subsection appendix Townsend} in the Appendix for the sake of a clearer
exposition. Things are different as far as the Euler class is concerned.
Indeed, we also prove in Subsection \ref{subsection appendix Townsend} that
\[
Q=\frac{1}{\pi^{q}}\int_{\mathbb{S}^{2q}}cw\left(  \operatorname*{Sym}%
(\operatorname*{Pf}),P_{\lambda_{\operatorname*{Id}}},\boldsymbol{\omega
}\right)  =2,
\]
which is obviously the Euler-Poincar\'{e} characteristic of the sphere
$\mathbb{S}^{2q}$. Since $SO(2q+1)\rightarrow\mathbb{S}^{2q}$ can be regarded
as the orthogonal frame bundle, this equality is simply a restatement of one
of the possible versions of the Gauss-Bonet Theorem (see \cite[page
112]{Dupont}).

If $q\geq3$, then, up to isomorphism, there only exists another principal
bundle structure over $\mathbb{S}^{2q}$, $\pi:P_{\delta}\rightarrow
\mathbb{S}^{2q}$, that given by the homomorphism $\delta:SO(2q)\rightarrow
SO(2q)$ introduced in Subsection \ref{examples field strength} \textbf{(iv)}.
In order to obtain the Euler class $\chi\left(  P_{\delta}\right)  $, one can
repeat the same computations carried out in Subsection
\ref{subsection appendix Townsend} in the Appendix just replacing
$\Omega_{\lambda_{\operatorname*{Id}}}^{\mathbf{\widetilde{W}}}$ with
$\Omega_{\delta}^{\mathbf{\widetilde{W}}}$. If we do so, it is not difficult
to realize that a $-1$ appears in each term of (\ref{monopole 29}) and,
therefore, $\int_{\mathbb{S}^{2q}}\chi\left(  P_{\delta}\right)  =-2$. In
other words, $P_{\lambda_{\operatorname*{Id}}}$ and $P_{\delta}$ have the same
charge with opposite sign. The details are left to the reader.

\subsection{$SU(2^{n-1})$-monopoles}

In this last example, we are going to review monopole configurations over
$\mathbb{S}^{2n}$, $n\in\mathbb{N}$, with gauge group $SU(2^{n-1})$. They have
been recently introduced in \cite{meng} by G. Meng as a generalization of Yang
and Dirac monopoles to higher dimensions. However, our approach to
$SU(2^{n-1})$-monopoles will differ from Meng's as we avoid referring to
spinor bundles. Instead, we will use the language of homogeneous principal
bundles developed so far. Hopefully, this may shed some light on the arguments
used in \cite{meng} to prove the existence of $SU(2^{n-1})$-monopoles, which
hence turn out to be part of a more general picture. Furthermore, this last
example suggests that the theory of representations of Lie groups is very
useful to give other monopole configurations.

Let $Spin(2n+1)\rightarrow Spin(2n+1)/Spin(2n)\cong\mathbb{S}^{2n}$ be the
canonical $Spin(2n)$-principal bundle over the sphere $\mathbb{S}^{2n}$,
$n\in\mathbb{N}$, and let $\boldsymbol{\omega}\in\Omega^{1}\left(
Spin(2n+1);\mathfrak{so}(2n)\right)  $ its canonical connection as in Example
\ref{example maurer-cartan}. It is a well known result that $Spin(2n)$ has two
different irreducible complex representations of dimension $2^{n-1}$. More
explicitly, there exist complex Hermitian vector spaces $V_{i}$, $i=1,2$, of
real dimension $2^{n-1}$ and a couple of non-equivalent homomorphisms%
\[
\lambda_{i}:Spin(2n)\longrightarrow Gl\left(  V_{i}\right)  ,~\text{~}i=1,2,
\]
such that $\lambda_{i}\left(  g\right)  $ leaves the Hermitian structure
invariant for any $g\in Spin(2n)$. In particular, this means that%
\[
\lambda_{i}\left(  Spin(2n)\right)  \subseteq SU(2^{n-1})
\]
and $\lambda_{i}$ can be considered as homomorphisms from $Spin(2n)$ to
$SU(2^{n-1})$. Therefore, there exist two distinct principal bundles
$\pi_{\lambda_{i}}:P_{\lambda_{i}}\rightarrow\mathbb{S}^{2n}$ with gauge group
$SU(2^{n-1})$ supporting a (left) $Spin(2n+1)$ action. Their canonical
connections $\boldsymbol{\omega}_{\lambda_{i}}\in\Omega^{1}(P_{\lambda_{i}%
};\mathfrak{so}(2n))$ are then induced from $\boldsymbol{\omega}$ according to
Proposition \ref{prop 1}. Moreover, they are $Spin(2n+1)$-invariant, and give
rise to the unique non-trivial $SU(2^{n-1})$-monopole configurations. Although
Meng does not prove in \cite{meng} that $\boldsymbol{\omega}_{\lambda_{i}}$
are Yang-Mills connections, they are so by Proposition
\ref{prop canonical is YM} indeed. For $n=1$ and $n=2$, these $SU(2^{n-1}%
)$-monopoles reduce to Dirac's and Yang's, respectively.

In \cite{meng}, the charge of these monopoles is also computed. It is proved
that%
\[
\frac{1}{n!}\int_{\mathbb{S}^{2n}}\operatorname*{trace}\left(  -\frac
{\Omega^{\boldsymbol{\omega}_{\lambda_{i}}}}{2\pi}\wedge\overset{n)}%
{...}\wedge-\frac{\Omega^{\boldsymbol{\omega}_{\lambda_{i}}}}{2\pi}\right)
=\left(  -1\right)  ^{i}.
\]
Which monopole has charge positive or negative depends on how we labelled the
homomorphisms $\lambda_{i}$, $i=1,2$.

\appendix

\section{Appendix}

\subsection{Proof of Proposition \ref{proposition YM}}

Before proving Proposition \ref{proposition YM}, we need an auxiliary lemma:

\begin{lemma}
\label{lema operador hodge}Let $\alpha\in\Omega^{k}\left(  \mathbb{S}%
^{2n}\right)  $ and $f:\mathbb{R}^{2n+1}\backslash\{0\}\rightarrow
\mathbb{S}^{2n}$ as in Equation (\ref{eq proyeccion homotopica}). If $r\in
C^{\infty}\left(  \mathbb{R}^{2n+1}\backslash\{0\}\right)  $ is the radius
function, $r\left(  x\right)  =\left\Vert x\right\Vert $, then%
\begin{equation}
\ast f^{\ast}(\alpha)=r^{2(n-k)}f^{\ast}\left(  \ast\alpha\right)  \wedge dr.
\label{monopole 32}%
\end{equation}

\end{lemma}

\begin{proof}
Let $y\in\mathbb{R}^{2n+1}\backslash\{0\}$ be an arbitrary point and
$z=y/\left\Vert y\right\Vert \in\mathbb{S}^{2n}$. We can take global Euclidean
coordinates $\left(  x^{1},...,x^{2n+1}\right)  $ on $\mathbb{R}%
^{2n+1}\backslash\{0\}$ such that $y=(0,\overset{2n)}{...},0,r(y))$. Then
$z=(0,\overset{2n)}{...},0,1)\in\mathbb{S}^{2n}\subset\mathbb{R}%
^{2n+1}\backslash\{0\}$. That is, $z$ can be regarded as the north pole of the
sphere $\mathbb{S}^{2n}$. The tangent space $T_{y}(\mathbb{R}^{2n+1}%
\backslash\{0\})$ can be decomposed as the direct sum%
\[
T_{y}(\mathbb{R}^{2n+1}\backslash\{0\})=T_{y}S_{r(y)}\oplus W_{y}%
\]
where $T_{y}S_{r(y)}$ is the tangent space to the sphere $S_{r(y)}$ of radius
$r(y)$ at $y$ and $W_{y}$ is its orthogonal complement, in the radial
direction. The first $2n$ coordinates $\left(  x^{1},...,x^{2n+1}\right)  $ we
have on $\mathbb{R}^{2n+1}\backslash\{0\}$ can be used around $z$ on
$\mathbb{S}^{2n}$ by means of the local diffeomorphism%
\[
\left(  x^{1},...,x^{2n}\right)  \longmapsto\left(  x^{1},...,x^{2n}%
,\sqrt{1-\sum\nolimits_{i=1}^{2n}\left(  x^{i}\right)  ^{2}}\right)  .
\]
Observe that the vector fields $\left\{  \frac{\partial}{\partial x^{1}%
},...,\frac{\partial}{\partial x^{2n}}\right\}  $ form an orthonormal basis at
$z\in\mathbb{S}^{2n}$ and that, as a vector space, $T_{y}S_{r(y)}$ is
isomorphic to $T_{z}\mathbb{S}^{2n}$. In this context, it is easy to see that%
\[
T_{y}f=\frac{1}{r(y)}\operatorname*{Id}\circ\left.  \operatorname*{proj}%
\right\vert _{T_{y}S_{r(y)}},
\]
where $\left.  \operatorname*{proj}\right\vert _{T_{y}S_{r(y)}}:T_{y}%
(\mathbb{R}^{2n+1}\backslash\{0\})\rightarrow T_{y}S_{r(y)}$ denotes the
projection onto $T_{y}S_{r(y)}$ and the isomorphism $T_{y}S_{r(y)}\cong
T_{z}\mathbb{S}^{2n}$ has been used. Consequently, if $\alpha$ is locally
written as $\sum_{i_{1}<...<i_{k}}\alpha_{i_{1}...i_{k}}dx^{i_{1}}%
\wedge...\wedge dx^{i_{k}}$ around $z$, it is immediate to see that%
\[
f^{\ast}\left(  \alpha\right)  \left(  y\right)  =\frac{1}{r^{k}(y)}%
\sum_{i_{1}<...<i_{k}}\alpha_{i_{1}...i_{k}}\left(  z\right)  \left(
dx^{i_{1}}\wedge...\wedge dx^{i_{k}}\right)  (y).
\]
On the other hand,%
\begin{align*}
(\ast\alpha)(z)  &  =\frac{1}{k!}\sum_{j_{1}<...<j_{2n-k}}\sum_{i_{1}%
<...<i_{k}}\alpha^{i_{1}...i_{k}}\left(  z\right)  \varepsilon_{i_{1}%
...i_{k}j_{1}...j_{2n-k}}^{1...2n}\left(  dx^{j_{1}}\wedge...\wedge
dx^{j_{2n-k}}\right)  (z),\\
f^{\ast}\left(  \ast\alpha\right)  (y)  &  =\frac{1}{k!r^{2n-k}(y)}\sum
_{j_{1}<...<j_{2n-k}}\sum_{i_{1}<...<i_{k}}\alpha^{i_{1}...i_{k}}\left(
z\right)  \varepsilon_{i_{1}...i_{k}j_{1}...j_{2n-k}}^{1...2n}\left(
dx^{j_{1}}\wedge...\wedge dx^{j_{2n-k}}\right)  (y),
\end{align*}
and%
\begin{equation}
\ast\left(  f^{\ast}\alpha\right)  (y)=\frac{1}{k!r^{k}(y)}\sum_{j_{1}%
<...<j_{2n+1-k}}\sum_{i_{1}<...<i_{k}}\alpha^{i_{1}...i_{k}}\left(  z\right)
\varepsilon_{i_{1}...i_{k}j_{1}...j_{2n+1-k}}^{1...2n+1}\left(  dx^{j_{1}%
}\wedge...\wedge dx^{j_{2n+1-k}}\right)  (y).\smallskip\label{monopole 33}%
\end{equation}
In these equations, $\varepsilon_{i_{1}...i_{k}j_{1}...j_{2n-k}}^{1...2n}$
denotes the totally antisymmetric symbol and $\alpha^{i_{1}...i_{k}}\left(
z\right)  =\alpha_{i_{1}...i_{k}}\left(  z\right)  $ because, in the
coordinates we chose, the matrix of the Euclidean metric is diagonal on both
$y\in\mathbb{R}^{2n+1}\backslash\{0\}$ and $z\in\mathbb{S}^{2n}$. Now, in each
non-zero term on the right hand side of (\ref{monopole 33}), the differential
$dx^{2n+1}$ appears. We can move it to the last right position just taking
into account a possible additional $\left(  -1\right)  ^{\left\vert
\sigma\right\vert }$ for a suitable permutation $\sigma$. This $\left(
-1\right)  ^{\left\vert \sigma\right\vert }$, however, cancels with the same
$\left(  -1\right)  ^{\left\vert \sigma\right\vert }$ that comes from moving
the index $j_{l}=2n+1$ in $\varepsilon_{i_{1}...i_{k}j_{1}...j_{2n+1-k}%
}^{1...2n+1}$ to the last right position. In this case, $\varepsilon
_{i_{1}...i_{k}j_{1}...j_{l}}^{1...2n+1}=\varepsilon_{i_{1}...i_{k}%
j_{1}...j_{2n-k}}^{1...2n}$. Therefore, (\ref{monopole 33}) equals%
\begin{gather*}
\frac{1}{k!r^{k}(y)}\left(  \sum_{j_{1}<...<j_{2n-k}}\sum_{i_{1}<...<i_{k}%
}\alpha^{i_{1}...i_{k}}\left(  z\right)  \varepsilon_{i_{1}...i_{k}%
j_{1}...j_{2n-k}}^{1...2n}\left(  dx^{j_{1}}\wedge...\wedge dx^{j_{2n-k}%
}\right)  (y)\right)  \wedge dx^{2n+1}(y)\smallskip\\
=r^{2(n-k)}(y)\left(  f^{\ast}\left(  \ast\alpha\right)  (y)\right)  \wedge
dx^{2n+1}(y).
\end{gather*}
Finally, observe that $dr$ coincides with $dx^{2n+1}$ at $y$, $dr\left(
y\right)  =dx^{2n+1}(y)$, so%
\[
\ast f^{\ast}(\alpha)(y)=r^{2(n-k)}(y)\left(  f^{\ast}\left(  \ast
\alpha\right)  (y)\right)  \wedge dr(y).
\]
Since the point $y\in\mathbb{R}^{2n+1}\backslash\{0\}$ we chose was completely
arbitrary, we conclude that (\ref{monopole 32}) holds globally.\smallskip
\end{proof}

\begin{proof}
[Proof of Proposition \ref{proposition YM}]Let $\varphi\in\Omega_{equiv}%
^{k}\left(  P;\mathfrak{g}\right)  ^{\operatorname*{Hor}}$ and let $F^{\ast
}:f^{\ast}\left(  P\right)  \rightarrow P$ be the natural bundle homomorphism
from the pull-back of $\pi:P\rightarrow S^{2n}$ by $f$ (see
\ref{eq proyeccion homotopica}). $F^{\ast}\left(  \varphi\right)  $ can be
naturally seen as a form in $\Omega_{equiv}^{k}\left(  f^{\ast}%
(P);\mathfrak{g}\right)  ^{\operatorname*{Hor}}$. It is not difficult to
realize then from Lemma \ref{lema operador hodge} that%
\[
\ast F^{\ast}\left(  \varphi\right)  =\overline{\pi}^{\ast}\left(
r^{2(n-k)}\right)  F^{\ast}\left(  \ast\varphi\right)  ~\bar{\wedge}%
~\overline{\pi}^{\ast}(dr),
\]
where $\overline{\pi}:f^{\ast}\left(  P\right)  \rightarrow\mathbb{R}%
^{2n+1}\backslash\{0\}$ and the product $\bar{\wedge}$ of two forms $\beta
\in\Omega^{r}\left(  f^{\ast}(P);\mathfrak{g}\right)  $ and $\alpha\in
\Omega^{q}\left(  f^{\ast}(P)\right)  $ must be understood through the product
of an element of the vector space $\mathfrak{g}$ by a real number; that is,%
\[
\beta~\bar{\wedge}~\alpha\left(  Y_{1},...,Y_{r+q}\right)  =\frac{1}{r!q!}%
\sum_{\sigma\in S_{r+q}}(-1)^{\left\vert \sigma\right\vert }\underset
{\in\mathbb{R}}{\underbrace{\alpha\left(  Y_{\sigma(1)},...,Y_{\sigma
(r)}\right)  }~}\underset{\in\mathfrak{g}}{\underbrace{\beta\left(
Y_{\sigma(r+1)},...,Y_{\sigma(r+q)}\right)  }},
\]
for any $\{Y_{1},...,Y_{r+q}\}\subset\mathfrak{X}\left(  f^{\ast}(P)\right)  $.

Let now $\omega\in\Omega_{equiv}^{1}\left(  P;\mathfrak{g}\right)  $ be a
principal connection and $F^{\ast}\left(  \omega\right)  \in\Omega_{equiv}%
^{1}\left(  f^{\ast}\left(  P\right)  ;\mathfrak{g}\right)  $ the
corresponding principal connection on $\overline{\pi}:f^{\ast}\left(
P\right)  \rightarrow\mathbb{R}^{2n+1}\backslash\{0\}$. Observe that%
\[
T_{y}F(\operatorname*{Hor}\nolimits_{y})=T_{y}F\left(  \ker\left(  F^{\ast
}(\omega)(y)\right)  \right)  =\ker\omega(F(y))=\operatorname*{Hor}%
\nolimits_{F(y)}\subset T_{F(y)}P
\]
therefore, as far as their field strengths is concerned (\cite[17.5]{michor}),%
\begin{align*}
\Omega^{F^{\ast}(\omega)}  &  =D^{F^{\ast}(\omega)}(F^{\ast}(\omega))=\left.
d\circ F^{\ast}(\omega)\right\vert _{\operatorname*{Hor}\nolimits_{y}}\\
&  =\left.  F^{\ast}(d\circ\omega)\right\vert _{\operatorname*{Hor}%
\nolimits_{y}}=F^{\ast}\left(  \left.  d\circ\omega\right\vert
_{\operatorname*{Hor}\nolimits_{F(y)}}\right)  =F^{\ast}\left(  \Omega
^{\omega}\right)  .
\end{align*}
Then,
\begin{align*}
-\delta^{F^{\ast}(\omega)}\Omega^{F^{\ast}(\omega)}  &  =\ast\circ D^{F^{\ast
}(\omega)}\circ\ast\left(  \Omega^{F^{\ast}(\omega)}\right)  =\ast\circ
D^{F^{\ast}(\omega)}\circ\ast\left(  F^{\ast}\left(  \Omega^{\omega}\right)
\right) \\
&  =\ast\circ D^{F^{\ast}(\omega)}\left(  \overline{\pi}^{\ast}\left(
r^{2(n-2)}\right)  F^{\ast}\left(  \ast\Omega^{\omega}\right)  \hspace
{2.16pt}\bar{\wedge}\hspace{2.16pt}\overline{\pi}^{\ast}(dr)\right) \\
&  =\ast\circ\left.  d\left(  \overline{\pi}^{\ast}\left(  r^{2(n-2)}\right)
F^{\ast}\left(  \ast\Omega^{\omega}\right)  \hspace{2.16pt}\bar{\wedge}%
\hspace{2.16pt}\overline{\pi}^{\ast}(dr)\right)  \right\vert
_{\operatorname*{Hor}}\\
&  =\ast\left(  \overline{\pi}^{\ast}\left(  r^{2(n-2)}\right)  F^{\ast
}\left(  \left.  d\circ\ast(\Omega^{\omega})\right\vert _{\operatorname*{Hor}%
}\right)  \hspace{2.16pt}\bar{\wedge}\hspace{2.16pt}\overline{\pi}^{\ast
}(dr)\right)
\end{align*}
where in the last line we have used that $\overline{\pi}^{\ast}(dr)$ was
already a horizontal form. Thus,%
\[
\delta^{F^{\ast}(\omega)}\Omega^{F^{\ast}(\omega)}=-\ast\left(  \overline{\pi
}^{\ast}\left(  r^{2(n-2)}\right)  F^{\ast}\left(  D^{\omega}\circ\ast
(\Omega^{\omega})\right)  \hspace{2.16pt}\bar{\wedge}\hspace{2.16pt}%
\overline{\pi}^{\ast}(dr)\right)
\]
Now, if $\alpha\in\Omega_{equiv}^{k}\left(  P;\mathfrak{g}\right)
^{\operatorname*{Hor}}$, then%
\begin{equation}
\ast\circ\ast\left(  \alpha\right)  =\left(  -1\right)  ^{k(m-k)}\alpha,
\label{monopole 35}%
\end{equation}
where $m=2n+1$ or $m=2n$ if the base manifold is $\mathbb{R}^{2n+1}%
\backslash\{0\}$ or $\mathbb{S}^{2n}$ respectively. On the other hand, by
Lemma \ref{lema operador hodge},%
\begin{align}
\ast\left(  F^{\ast}\left(  \ast\circ D^{\omega}\circ\ast(\Omega^{\omega
})\right)  \right)   &  =\overline{\pi}^{\ast}\left(  r^{2(n-1)}\right)
F^{\ast}\left(  \ast\circ\ast\circ D^{\omega}\circ\ast(\Omega^{\omega
})\right)  \overline{\pi}^{\ast}\left(  dr\right) \nonumber\\
&  =\left(  -1\right)  ^{2n-1}\overline{\pi}^{\ast}\left(  r^{2(n-1)}\right)
F^{\ast}\left(  D^{\omega}\circ\ast(\Omega^{\omega})\right)  \hspace
{2.16pt}\bar{\wedge}\hspace{2.16pt}\overline{\pi}^{\ast}\left(  dr\right)  .
\label{monopole 34}%
\end{align}
Taking the Hodge operator in both sides of (\ref{monopole 34}) and using
(\ref{monopole 35}),%
\[
\left(  F^{\ast}\left(  \ast\circ D^{\omega}\circ\ast(\Omega^{\omega})\right)
\right)  =-\ast\left(  \overline{\pi}^{\ast}\left(  r^{2(n-1)}\right)
F^{\ast}\left(  D^{\omega}\circ\ast(\Omega^{\omega})\right)  \hspace
{2.16pt}\bar{\wedge}\hspace{2.16pt}\overline{\pi}^{\ast}\left(  dr\right)
\right)  ,
\]
so%
\[
\delta^{F^{\ast}(\omega)}\Omega^{F^{\ast}(\omega)}=\overline{\pi}^{\ast
}\left(  \frac{r^{2(n-2)}}{r^{2(n-1)}}\right)  F^{\ast}\left(  \ast\circ
D^{\omega}\circ\ast(\Omega^{\omega})\right)  =-\frac{1}{\overline{\pi}^{\ast
}\left(  r^{2}\right)  }F^{\ast}\left(  \delta^{\omega}\Omega^{\omega}\right)
.
\]

\end{proof}

\subsection{Proof of Proposition \ref{prop 1}\label{appendix prop 1}}

Sometimes, principal connections are more conveniently described by means of a
one form $\Phi\in\Omega^{1}\left(  P;VP\right)  $ with values on the vertical
bundle $VP=\cup_{p\in P}\operatorname*{Ver}_{p}$,%
\[
\Phi_{p}\left(  X\right)  =T_{e}R_{p}\circ\omega_{p}(X).
\]
In this expression $X\in\mathfrak{X}(P)$, $p\in P$, $e\in G$ denotes the unit
element, and $R_{p}:G\rightarrow P$ is the right action $R_{p}(g):=R(g,p)$ for
any $g\in G$. The principal connection $\Phi$ satisfies that $TR_{g}\circ
\Phi=\Phi\circ TR_{g}$ or, equivalently, $\Phi=TR_{g^{-1}}\circ\Phi\circ
TR_{g}$ for any $g\in G$.

In the particular case of homogeneous principal bundles $\pi:P_{\lambda
}\rightarrow K/H$, principal connections $\Phi_{\lambda}\in\Omega^{1}\left(
P_{\lambda};VP_{\lambda}\right)  $ can be built from principal connections
$\Phi\in\Omega^{1}\left(  K;VK\right)  $ on $K\rightarrow K/H$. In order to
show how this construction works, we are going to explicitly describe
$TP_{\lambda}$. First of all, it can be proved that $T\pi:TK\rightarrow
T(K/H)$ is again a principal bundle with structural group $TH$ with right
action,%
\begin{equation}%
\begin{array}
[c]{rrl}%
TR:TK\times TH & \longrightarrow & TK\\
\left(  \left(  k,X_{k}\right)  ,\left(  h,X_{h}\right)  \right)  &
\longmapsto & \left(  kh,T_{h}L_{k}(X_{h})+T_{k}R_{h}(X_{k})\right)  ,
\end{array}
\label{monopole 5}%
\end{equation}
where $X_{h}\in T_{h}H$ and $X_{k}\in T_{k}K$. In addition, if
$\operatorname*{inv}:H\rightarrow H$, $\operatorname*{inv}(h):=h^{-1}$ denotes
the inverse map of the Lie group $H$, $TH$ acts on $TG$ by the right action%
\begin{equation}%
\begin{array}
[c]{rrl}%
TG\times TH & \longrightarrow & TK\\
\left(  \left(  g,X_{g}\right)  ,\left(  h,X_{h}\right)  \right)  &
\longmapsto & \left(  \lambda(h)^{-1}g,T_{g}L_{\lambda(h)^{-1}}(X_{g}%
)+T_{\lambda(h)^{-1}}R_{g}\circ T_{h^{-1}}\lambda\circ T_{h}%
\operatorname*{inv}(X_{h})\right)  ,
\end{array}
\label{monopole 6}%
\end{equation}
so that the tangent space $TP_{\lambda}$ equals the associated bundle
$TK\times_{TH}TG$ (\cite[Theorem 18.18]{michor}). That is, $TP_{\lambda}$ is
the orbit space of $TK\times TG$ under the $TH$-action $T\Psi_{\lambda}$.
Using the fact that $TP_{\lambda}=TK\times_{TH}TG$, the connection
$\Phi_{\lambda}$ induced from $\Phi$ is defined by the following commutative
diagram:%
\begin{equation}%
\begin{array}
[c]{rrl}%
TK\times TG & \overset{\Phi\times\operatorname*{Id}}{\longrightarrow} &
TK\times TG\\
_{Tq}\downarrow &  & \downarrow~_{Tq}\smallskip\\
TK\times_{TH}TG & \underset{\Phi_{\lambda}}{\longrightarrow} & TK\times
_{TH}TG=T\left(  K\times_{H}G\right)  ,
\end{array}
\label{monopole 3}%
\end{equation}
where $q:K\times G\rightarrow K\times_{H}G$ sends each element to its
corresponding equivalent class in $K\times_{H}G$ and $Tq$ is its tangent
map.\smallskip

\begin{proof}
[Proof of Proposition \ref{prop 1}]Take $p=[e,e]^{\sim}\in P_{\lambda}$ on the
fiber $\pi^{-1}\left(  o\right)  $ and let $\xi\in\mathfrak{k}$. Since the
$K$-action on $P_{\lambda}$ is simply the left action $L_{\lambda}$ introduced
in (\ref{monopole 2}), the infinitesimal generator $\xi_{P_{\lambda}}$ at $p$
corresponds to the equivalent class $[\xi,0]_{p}^{\sim}$ in $TK\times_{TH}TG$.
Observe that $[\xi,0]_{p}^{\sim}$ denotes the orbit of $\left(  (e,\xi
),(e,0)\right)  \in TK\times TG$ under the action of $TH$. By
(\ref{monopole 5}) and (\ref{monopole 6}), $[\xi,0]_{p}^{\sim}$ is equivalent
to%
\[
\left[  X_{h}+T_{e}R_{h}(\xi),T_{h^{-1}}\lambda\circ T_{h}\operatorname*{inv}%
(X_{h})\right]  _{(h,\lambda(h)^{-1})}^{\sim}%
\]
for any $X_{h}\in T_{h}H$, $h\in H$. Taking $h=e\in H$, we have%
\begin{equation}
\lbrack\xi,0]_{p}^{\sim}=\left[  \eta+\xi,-T_{e}\lambda(\eta)\right]
_{p}^{\sim}. \label{monopole 7}%
\end{equation}
In (\ref{monopole 7}), we have written $\eta\in\mathfrak{h}=T_{e}H$ instead of
$X_{e}$ and have used $T_{e}\operatorname*{inv}=-\operatorname*{Id}$.

On the other hand, $\Phi\in\Omega^{1}\left(  K;VK\right)  $ coincides with the
projection $\operatorname*{proj}_{\mathfrak{h}}:\mathfrak{k}\rightarrow
\mathfrak{h}$ from $\mathfrak{k}$ to $\mathfrak{h}$ at $e\in K$. Therefore,
(\ref{monopole 3}) implies
\[
(\Phi_{\lambda})\left(  [e,e]^{\sim}\right)  (\xi_{P_{\lambda}})=(\Phi
_{\lambda})\left(  [e,e]^{\sim}\right)  \left(  [\xi,0]_{(e,e)}^{\sim}\right)
=[\operatorname*{proj}\nolimits_{\mathfrak{h}}(\xi),0]_{(e,e)}^{\sim}.
\]
By (\ref{monopole 7}) with $\eta=-\operatorname*{proj}\nolimits_{\mathfrak{h}%
}(\xi)$, $[\operatorname*{proj}\nolimits_{\mathfrak{h}}(\xi),0]_{(e,e)}^{\sim
}$ is equivalent to%
\[
\left[  0,T_{e}\lambda\left(  \operatorname*{proj}\nolimits_{\mathfrak{h}}%
(\xi)\right)  \right]  _{(e,e)}^{\sim}.
\]
Now, for any $\eta\in\mathfrak{g}$, $T_{e}(R_{\lambda})_{p}(\eta)\in
TP_{\lambda}$ equals $[0,\eta]_{p}^{\sim}$ in $TK\times_{TH}TG$. Hence, the
principal connection $\omega_{p}=(T_{e}(R_{\lambda})_{p})^{-1}\circ
(\Phi_{\lambda})_{p}$ satisfies%
\begin{align*}
\omega\left(  p\right)  (\xi_{P_{\lambda}})  &  =(T_{e}(R_{\lambda})_{p}%
)^{-1}\left(  [\operatorname*{proj}\nolimits_{\mathfrak{h}}(\xi),0]_{(e,e)}%
^{\sim}\right)  =(T_{e}(R_{\lambda})_{p})^{-1}\left(  \left[  0,T_{e}%
\lambda\left(  \operatorname*{proj}\nolimits_{\mathfrak{h}}(\xi)\right)
\right]  _{(e,e)}^{\sim}\right) \\
&  =T_{e}\lambda\left(  \operatorname*{proj}\nolimits_{\mathfrak{h}}%
(\xi)\right)  =\mathbf{W}\left(  \xi\right)
\end{align*}
if $\mathbf{W}:\mathfrak{k}\rightarrow\mathfrak{g}$ is the canonical connection.
\end{proof}

\subsection{Characteristic classes of
SO(2n)-monopoles\label{subsection appendix Townsend}}

According to what we said in Subsection \ref{subsection Townsend monopole}, we
are going to explicitly show that the $q$-th Chern class of the principal
bundle $SO(2q+1)\rightarrow\mathbb{S}^{2q}=SO(2q+1)/SO(2q)$ is zero. If $q$ is
odd, then the characteristic coefficient $c_{q}^{2q}$ is zero and,
consequently, so is the corresponding $q$-th Chern class. If $q$ is even then,
by (\ref{monopole 30}),
\begin{align*}
&  \left(  -1\right)  ^{q}\left(  2\pi i\right)  ^{q}q!\,\pi^{\ast}\left(
cw\left(  \operatorname*{Sym}(c_{q}^{2q}),P,\boldsymbol{\omega}\right)
\right)  \left(  \upsilon_{1},...,\upsilon_{2q}\right) \\
&  =\sum_{i_{1}<...<i_{q}}\sum_{\eta\in S_{q}}\left(  -1\right)  ^{\left\vert
\eta\right\vert }\left(  \Omega^{\boldsymbol{\omega}}\right)  _{\eta(i_{1}%
)}^{i_{1}}\wedge...\wedge\left(  \Omega^{\boldsymbol{\omega}}\right)
_{\eta(i_{q})}^{i_{q}}\left(  \upsilon_{1},...,\upsilon_{2q}\right) \\
&  =\frac{1}{2^{q}}\sum_{i_{1}<...<i_{q}}\sum_{\eta\in S_{q}}\left(
-1\right)  ^{\left\vert \eta\right\vert }\sum_{\sigma\in S_{2q}}\left(
-1\right)  ^{\left\vert \sigma\right\vert }\left(  \Omega_{\lambda
_{\operatorname*{Id}}}^{\mathbf{\widetilde{W}}}\left(  \upsilon_{\sigma
(1)},\upsilon_{\sigma(2)}\right)  \right)  _{\eta(i_{1})}^{i_{1}}\cdots\left(
\Omega_{\lambda_{\operatorname*{Id}}}^{\mathbf{\widetilde{W}}}\left(
\upsilon_{\sigma(2q-1)},\upsilon_{\sigma(2q)}\right)  \right)  _{\eta(i_{q}%
)}^{i_{q}}\\
&  =\frac{1}{2^{q}}\sum_{i_{1}<...<i_{q}}\sum_{\eta\in S_{q}}\left(
-1\right)  ^{\left\vert \eta\right\vert }\sum_{\sigma\in S_{2q}}\left(
-1\right)  ^{\left\vert \sigma\right\vert }\left(  \xi_{\sigma(1),\sigma
(2)}\right)  _{\eta(i_{1})}^{i_{1}}\cdots\left(  \xi_{\sigma(2q-1),\sigma
(2q)}\right)  _{\eta(i_{q})}^{i_{q}}.
\end{align*}
Using (\ref{monopole 31}), $\left(  \xi_{\sigma(1),\sigma(2)}\right)
_{\eta(i_{1})}^{i_{1}}\cdots\left(  \xi_{\sigma(2q-1),\sigma(2q)}\right)
_{\eta(i_{q})}^{i_{q}}$ equals%
\[
\prod_{r\in\{1,2,...,q\}}\left(  -1\right)  ^{U(\sigma(2r)-\sigma
(2r-1))}\left(  -1\right)  ^{U(\eta(i_{r})-i_{r})}\delta_{\sigma(2r-1)}%
^{i_{r}}\delta_{\sigma(2r)\eta(i_{r})}.
\]
But $\delta_{\sigma(2r-1)}^{i_{r}}\delta_{\sigma(2r)\eta(i_{r})}$ must be zero
for some $r\in\{1,2,...,q\}$ for any $\sigma\in S_{2q}$ because $\{\sigma
(2r-1),\sigma(2r)\}$ cover all the indices in $\{1,2,....,2q\}$ as $r$ ranges
from $1$ to $q$ but $\{i_{r},\eta(i_{r})\}$ only $q$ of them. Thus,%
\[
\pi^{\ast}\left(  cw\left(  \operatorname*{Sym}(c_{q}^{2q}%
),P,\boldsymbol{\omega}\right)  \right)  \left(  \upsilon_{1},...,\upsilon
_{2q}\right)  =0
\]
and the Chern class vanishes.

The same argument applied to the Euler class (\ref{monopole 28}) shows that%
\begin{align}
&  2^{q}q!\,\pi^{\ast}\left(  cw\left(  \operatorname*{Sym}(\operatorname*{Pf}%
),P_{\lambda_{\operatorname*{Id}}},\boldsymbol{\omega}\right)  \right)
\left(  \upsilon_{1},...,\upsilon_{2q}\right)  \nonumber\\
&  =\left(  \sum_{\eta\in S_{2q}}\left(  -1\right)  ^{\left\vert
\eta\right\vert }\left(  \Omega^{\boldsymbol{\omega}}\right)  _{\eta(2)}%
^{\eta(1)}\wedge\cdots\wedge\left(  \Omega^{\boldsymbol{\omega}}\right)
_{\eta(2q)}^{\eta(2q-1)}\right)  \left(  \upsilon_{1},...,\upsilon
_{2q}\right)  \nonumber\\
&  =\frac{1}{2^{q}}\sum_{\eta\in S_{2q}}\left(  -1\right)  ^{\left\vert
\eta\right\vert }\sum_{\sigma\in S_{2q}}\left(  -1\right)  ^{\left\vert
\sigma\right\vert }\left(  \Omega_{\lambda_{\operatorname*{Id}}}%
^{\mathbf{\widetilde{W}}}\left(  \upsilon_{\sigma(1)},\upsilon_{\sigma
(2)}\right)  \right)  _{\eta(2)}^{\eta(1)}\cdots\left(  \Omega_{\lambda
_{\operatorname*{Id}}}^{\mathbf{\widetilde{W}}}\left(  \upsilon_{\sigma
(2q-1)},\upsilon_{\sigma(2q)}\right)  \right)  _{\eta(2q)}^{\eta
(2q-1)}\nonumber\\
&  =\frac{1}{2^{q}}\sum_{\eta\in S_{2q}}\left(  -1\right)  ^{\left\vert
\eta\right\vert }\sum_{\sigma\in S_{2q}}\left(  -1\right)  ^{\left\vert
\sigma\right\vert }\left(  \xi_{\sigma(1),\sigma(2)}\right)  _{\eta(2)}%
^{\eta(1)}\cdots\left(  \xi_{\sigma(2q-1),\sigma(2q)}\right)  _{\eta
(2q)}^{\eta(2q-1)}.\label{monopole 29}%
\end{align}
Using (\ref{monopole 31}), (\ref{monopole 29}) equals%
\begin{align*}
&  \frac{1}{2^{q}}\sum_{\eta\in S_{2q}}\left(  -1\right)  ^{\left\vert
\eta\right\vert }\sum_{\sigma\in S_{2q}}\left(  -1\right)  ^{\left\vert
\sigma\right\vert }\prod_{i\in\{1,3,...,2q-1\}}\left(  -1\right)
^{U(\sigma(i+1)-\sigma(i))}\left(  -1\right)  ^{U(\eta(i+1)-\eta(i))}%
\delta_{\sigma(i)}^{\eta(i)}\delta_{\sigma(i+1)\eta(i+1)}\\
&  =\frac{1}{2^{q}}\sum_{\eta\in S_{2q}}\left(  -1\right)  ^{\left\vert
\eta\right\vert }\left(  -1\right)  ^{\left\vert \eta\right\vert }%
=\frac{\left(  2q\right)  !}{2^{q}}.
\end{align*}
Since $\operatorname{vol}\left(  \mathbb{S}^{2q}\right)  =\frac{2^{2q+1}%
\pi^{q}q!}{(2q)!}$, we conclude that the charge $Q$ of the monopole is%
\[
Q=\frac{1}{\pi^{q}}\int_{\mathbb{S}^{2q}}cw\left(  \operatorname*{Sym}%
(\operatorname*{Pf}),P_{\lambda_{\operatorname*{Id}}},\boldsymbol{\omega
}\right)  =\frac{1}{2^{2q}\pi^{q}q!}\left(  2q\right)  !\operatorname{vol}%
\left(  \mathbb{S}^{2q}\right)  =2
\]

\bigskip

\noindent\textbf{Acknowledgments.} The authors are indebted to Manuel Asorey,
Luis Joaqu\'{\i}n Boya, Adil Belhaj, Gregory Naber, Antonio Segu\'{\i}, and
Paul K. Townsend for their valuable comments and suggestions. They also thank
Juan-Pablo Ortega for his critical reading of the manuscript. J.-A. L.-C.
acknowledges partial support from MEC grant BFM2006-10531 and Gobierno de
Arag\'{o}n grant DGA-Grupos Consolidados 225-206. P. D. has been supported by
Grupo Te\'{o}rico de Altas Energ\'{\i}as (Grupo de Excelencia 2008 DGA) and
MEC grant FPA2006-02315.


\begin{thebibliography}{99999}                                                                                            %


\bibitem[ACO83]{asorey}Asorey, M., Cari\~{n}ena, J. F., and del Olmo, M. A.
Vector bundle representations of groups in quantum physics. \textit{J. Phys.
A} \textbf{16} (8), pp. 1603--1609 (1983).

\bibitem[B81]{bleecker}Bleecker, D. \textit{Gauge Theory and Variational
Principles}. Addison-Wesley publishing company, inc. (1981).

\bibitem[CE48]{chevalley cohomology}Chevalley, C. and Eilenberg, S. Cohomology
theory of Lie groups and Lie algebras. \textit{Trans. Amer. Math. Soc.}
\textbf{63} (1), pp. 85--124 (1948).

\bibitem[DV80]{viallet}Daniel, M. and Viallet, C. M. The geometrical setting
of gauge theories of Yang-Mills type. \textit{Reviews of Modern Physics}
\textbf{54} (2), pp. 175-197 (1980).

\bibitem[D31]{Dirac}Dirac, P. A. M. \textit{Proc. Roy. Soc. A} \textbf{133},
60 (1931).

\bibitem[D78]{Dupont}Dupont J. L. \textit{Curvature and Characteristic
Classes.} Lecture Notes in Mathematics \textbf{640}. Springer-Verlag (1978).

\bibitem[EGH80]{Eguchi}Eguchi, T., Gilkey, P. B., and Hanson, A. J.
Gravitation, gauge theories and differential geometry. \textit{Physics
Reports} \textbf{66} (6), pp. 213-393 (1980).

\bibitem[GT06]{townsend}Gibbons, G. W. and Townsend, P. K. Self-gravitating
Yang monopoles in all dimensions. \textit{Class. Quantum Grav.} \textbf{23},
pp. 4873--4885 (2006).

\bibitem[HH03]{hambelton}Hambleton, I. and Hausmann, J.-C. Equivariant
principal bundles over spheres and cohomogeneity one manifolds. \textit{Proc.
London Math. Soc. (3)} \textbf{86} (1), pp. 250--272 (2003).

\bibitem[H78]{helgason}Helgason, S. \textit{Differential Geometry, Lie Groups,
and Symmetric Spaces.} Pure and Applied Mathematics. Academic Press. (1978).

\bibitem[HST80]{harnard2}Harnard, J., Shinder, S., and Tafel, J. Canonical
connections on Riemannian symmetric spaces and solutions to the
Einstein-Yang-Mills equations. \textit{J. Math. Phys}. \textbf{21} (8), pp.
2236-2240 (1980).

\bibitem[HSV80]{harnard}J. Harnard, S. Shinder, and L. Vinet. Group actions on
principal bundles and invariance conditions for gauge fields. \textit{J. Math.
Phys.} \textbf{21} (12), pp. 2719-2724 (1980).

\bibitem[I89]{isham}C. J. Isham. \textit{Modern Differential Geometry for
Physicists}. World Scientific Lecture Notes in Physics Vol. \textbf{32}. World
Scientific, 1989.

\bibitem[I81]{Itoh}Itoh, M. Invariant connections and Yang-Mills solutions.
\textit{Transactions of the American Mathematical Society} \textbf{267} (1),
pp. 229-236 (1981).

\bibitem[KR74]{kalb-ramond}Kalb, M. and P. Ramond, P. \textit{Phys. Rev. D}
\textbf{9}, p. 2273 (1974).

\bibitem[KN69a]{kobayasi}Kobayashi, S. and Nominzu, K. \textit{Foundations of
Differential Geometry}. Volume I. Interscience Publishers (1969).

\bibitem[KN69b]{kobayashi 2}Kobayashi, S. and Nominzu, K. \textit{Foundations
of Differential Geometry}. Volume II. Interscience Publishers (1969).

\bibitem[L92]{laquer}Laquer, H. T. Invariant affine connections on symmetric
spaces. \textit{Proceedings of the American Mathematical Society} \textbf{155}
(2), pp. 447-454 (1992).

\bibitem[M07]{meng}Meng, G. Dirac and Yang Monopoles revisited.
\textit{Central European Journal of Physics} \textbf{5} (4), pp. 570-575
(2007). \texttt{math-ph/040905}.

\bibitem[M07]{michor}Michor, P. \textit{Topics in Differential Geometry}.
Springer. On-line version availabe at
\texttt{http://www.mat.univie.ac.at/\symbol{126}michor/dgbook.html} (2007).

\bibitem[M56a]{milnor-a}Milnor, J. W. On the construction of universal
bundles, \textit{Annals of Math.} \textbf{63} (2), p. 272 (1956).

\bibitem[M56b]{milnor-b}Milnor, J. W. On the construction of universal bundles
II, \textit{Annals of Math.} \textbf{63} (2), pp. 430-436 (1956).

\bibitem[M01]{morrison}S. Morrison. An introduction to pull-backs of bundles
and homotopy invariance. \texttt{arXiv:math/0105161v1.}

\bibitem[N97]{naber foundations}Naber, G. L. \textit{Topology, Geometry, and
Gauge Fields. Foundations}. Texts in Applied Mathematics \textbf{25}. Springer (1997).

\bibitem[N00]{naber}Naber, G. L. \textit{Topology, Geometry, and Gauge Fields.
Interactions}. Applied Mathematical Sciences \textbf{141}. Springer (2000).

\bibitem[N85]{nepomechie}Nepomechie, R. \textit{Phys. Rev. D} \textbf{31}, p.
1921 (1985).

\bibitem[S51]{steenrod}Steenrod, N. \textit{The Topology of Fiber Bundles}.
Princeton Math. Series \textbf{14}. Princeton University Press (1951).

\bibitem[T08]{tchrakian}Tchrakian, T. Dirac-Yang monopoles in all dimensions
and their regular counterparts. \textit{Physics of Atomic Nuclei} \textbf{71}
(6), pp.1116-1122 (2008).

\bibitem[TZ00]{tchrakian 2000}Tchrakian, D. H. and Zimmerschied, F. 't Hooft
tensors as Kalb-Ramond fields of generalized monopoles in all odd dimensions:
d=3 and d=5. \textit{Physical Review D} \textbf{62}, 045002 (2000).

\bibitem[T08]{the}The, D. Invariant Yang-Mills connections over non-reductive
pseudo-Riemannian homogeneous spaces. To appear in \textit{Transactions of the
AMS}. Availabe at \texttt{math-ph/0710.1865 }(2008).

\bibitem[Y78]{yang}Yang, C. N. Generalization of Dirac's monopole to $SU_{2}$
gauge fields. \textit{J. Math. Phys.} \textbf{19} (1), pp. 320-328 (1978).

\bibitem[W58]{wang}Wang, H-C. On invariant connections over a principal fiber
bundle. \textit{Nagoya Math. J.} \textbf{13}, pp. 1-19 (1958).
\end{thebibliography}
\end{document}